\documentclass[a4paper, 11pt]{article}

\usepackage{amsmath, amsthm, amssymb}
\usepackage{myMath}
\usepackage{enumitem}
\usepackage{graphicx}
\usepackage[caption=false]{subfig}
\usepackage{stackrel}
\usepackage{hyperref}

\title{Definability Equals Recognizability for $k$-Outerplanar
Graphs\footnote{The research was done while the first author was a
student at Utrecht University.}\footnote{The research of the second
author was partially funded by the Networks programme, funded by the Dutch Ministry of Education, Culture and
Science through the Netherlands Organisation for Scientific
Research.}
}

\author{Lars Jaffke\protect\footnote{Centrum Wiskunde \& Informatica.
Postbus 94079, 1090 GB Amsterdam. Email: \texttt{l.jaffke@cwi.nl}}
\and 
Hans L. Bodlaender\protect\footnote{Department of
Information and Computing Sciences, Utrecht University, P.O. Box 80.089, 3508 TB
Utrecht, The Netherlands.
Department of Mathematics and Computer Science,
University of Technology Eindhoven,
P.O. Box 513, 5600 MB Eindhoven, The Netherlands.
Email: \texttt{h.l.bodlaender@uu.nl}}}

\theoremstyle{theorem}
\newtheorem{theorem}{Theorem}[section]
\newtheorem{lemma}[theorem]{Lemma}
\newtheorem{proposition}[theorem]{Proposition}
\newtheorem{corollary}[theorem]{Corollary}
\theoremstyle{definition}
\newtheorem{definition}[theorem]{Definition}
\theoremstyle{remark}
\newtheorem{remark}[theorem]{Remark}
\newtheorem{example}[theorem]{Example}

\date{}

\begin{document}
	
	\maketitle
	
	\begin{abstract}
		One of the most famous algorithmic meta-theorems states that every graph
		property that can be defined by a sentence in counting monadic second order
		logic (CMSOL) can be checked in linear time for graphs of bounded treewidth,
		which is known as Courcelle's Theorem \cite{Cou90}.
		These algorithms are constructed as finite state tree automata, and hence every
		CMSOL-definable graph property is recognizable. Courcelle also conjectured that
		the converse holds, i.e.\ every recognizable graph property is definable in
		CMSOL for graphs of bounded treewidth. We prove this conjecture for
		$k$-outerplanar graphs, which are known to have treewidth at most $3k-1$
		\cite{Bod98}.
	\end{abstract}
	
	\section{Introduction}\label{secIntro}
A seminal result from 1990 by Courcelle states that
for every graph property $P$ that can be formulated in a language called
counting monadic second order logic (CMSOL), and each fixed $k$, there is a
linear time algorithm that decides $P$ for a graph given a tree decomposition of
width at most $k$ \cite{Cou90} (while similar results were discovered by Arnborg
et al.
\cite{ALS91} and Borie et al. \cite{BPT92}).
Counting monadic second order logic generalizes monadic second order logic
(MSOL) with a collection of predicates testing the size of sets modulo
constants. Courcelle showed that this makes the logic strictly more powerful
\cite{Cou90}.
The algorithms constructed in Courcelle's proof have the shape of a finite state
tree automaton and hence we can say that CMSOL-definable graph properties are
recognizable (or, equivalently, regular or finite-state).
Courcelle's Theorem generalizes one direction of a classic result in automata
theory by B{\"u}chi, which states that a language is recognizable,
if and only if it is MSOL-definable \cite{Bue60}.
Courcelle conjectured in 1990 that the other direction of B{\"u}chi's result can
also be generalized for graphs of bounded treewidth in CMSOL, i.e.\ that each
recognizable graph property is CMSOL-definable. \par
This conjecture is still regarded to be open. Its claimed resolution by
Lapoire \cite{Lap98} is not considered to be valid by several experts.
In the course of time proofs were given for the classes of trees and forests 
\cite{Cou90}, partial 2-trees \cite{Cou91}, partial 3-trees and $k$-connected
partial $k$-trees \cite{Kal00}. A sketch of a proof for graphs of pathwidth
at most $k$ appeared at ICALP 1997 \cite{Kab97}. Very recently, one of the
authors proved, in collaboration with Heggernes and Telle, that Courcelle's
Conjecture holds for partial $k$-trees without chordless cycles of length at
least $\ell$ \cite{BHT15}. 
\par
By the results presented in this paper, we add the class of $k$-outerplanar
graphs to this list. In particular, we first prove the conjecture for
3-connected $k$-outerplanar graphs and then generalize this result to all
$k$-outerplanar graphs, based on the decomposition of a connected graph into its
3-connected components, discovered by Tutte \cite{Tut66} and shown to be
definable in monadic second order logic by Courcelle \cite{Cou99}.
\par
The rest of the paper is organized as follows. In Section \ref{secPrel} we give
the basic definitions and review the concepts involved in our proofs. We present
the main result in Section \ref{secKOPG} and conclude in Section \ref{secConc}.

\section{Preliminaries}\label{secPrel}
\subsection{Graphs and Tree Decompositions}\label{secPrelGTD}
Throughout the paper, a graph $G = (V, E)$ with vertex set $V$ and edge set $E$
is undirected, connected and simple. We denote the subgraph
relation by $G \sqsubseteq H$ and for a set $W \subseteq V$, $G[W]$ denotes the induced
subgraph over $W$ in $G$, so $G[W] = (W, E \cap (W \times W))$. We call a set
$C \subset V$ a \emph{cut} of $G$, if $G[V \setminus C]$ is disconnected. An
\emph{$\ell$-cut} of $G$ is a cut of size $\ell$. A set $S \subseteq V$ is
said to be \emph{incident} to an $\ell$-cut $C$, if $C \subset S$. We call a
graph \emph{$\ell$-connected}, if it does not contain a cut of size at most $\ell-1$.
\par
We now define the class of $k$-outerplanar graphs and some central notions used
extensively throughout the rest of the paper.
\begin{definition}[(Planar) Embedding]
	A drawing of a graph in the plane is called an \emph{embedding}. If no pair of
	edges in this drawing crosses, then it is called \emph{planar}.
\end{definition}
\begin{definition}[$k$-Outerplanar Graph]\label{defKOPG}
	Let $G = (V, E)$ be a graph. $G$ is called a \emph{planar
	graph}, if there exists a planar embedding of $G$.
	An embedding of a graph $G$ is \emph{$1$-outerplanar}, if it is
	planar, and all vertices lie on the exterior face. For $k \ge 2$, an embedding
	of a graph $G$ is \emph{$k$-outerplanar}, if it is planar, and when
	all vertices on the outer face are deleted, then one obtains a
	$(k-1)$-outerplanar embedding of the resulting graph. If $G$ admits a
	$k$-outerplanar embedding, then it is called a \emph{$k$-outerplanar graph}.
\end{definition}
The following definition will play a central role in many of the proofs of
Section \ref{secKOPG}.
\begin{definition}[Fundamental Cycle]\label{defFundCyc}
	Let $G = (V, E)$ be a graph with maximal spanning forest $T = (V, F)$. Given an
	edge $e = \{v, w\}$, $e \in E \setminus F$, its \emph{fundamental cycle} is a
	cycle that is formed by the unique path from $v$ to $w$ in $F$ together with
	the edge $e$.
\end{definition}
\begin{definition}[Tree Decomposition, Treewidth]
	A \emph{tree decomposition} of a graph $G = (V, E)$ is a pair $(T, X)$ of a
	tree $T = (N, F)$ and an indexed family of vertex sets $(X_t)_{t \in N}$
	(called \emph{bags}), such that the following properties hold.
	\begin{enumerate}[label=(\roman*)]
	  \item Each vertex $v \in V$ is contained in at least one bag.
	  \item For each edge $e \in E$ there exists a bag containing both endpoints.
	  \item For each vertex $v \in V$, the bags in the tree decomposition that
	  contain $v$ form a subtree of $T$.
	\end{enumerate}
	The \emph{width} of a tree decomposition is the size of the largest bag minus 1
	and the \emph{treewidth} of a graph is the minimum width of all its tree
	decompositions. We might sometimes refer to graphs of treewidth at most $k$ as
	\emph{partial $k$-trees}.\footnote{For several characterizations of graphs of
	treewidth at most $k$, see e.g.\ \cite[Theorem 1]{Bod98}}
\end{definition}
To avoid confusion, in the following we will refer to elements of $N$ as
\emph{nodes} and elements of $V$ as \emph{vertices}. Sometimes, to shorten the
notation, we might not differ between the terms \emph{node} and \emph{bag} in a
tree decomposition. \par
We use the following notation. If $P$ denotes a graph property (e.g.\ a graph
has a Hamiltonian cycle), then by '$P(G)$' we express that a graph $G$ has
property $P$.

\subsection{Monadic Second Order Logic of Graphs}\label{secPrelDef}
 We now define counting monadic second order logic of graphs $G = (V, E)$,
 using terminology from \cite{BPT92} and \cite{Kal00}.
 Variables in this predicate logic are either single vertices/edges or
 vertex/edge sets. We form predicates by joining \emph{atomic predicates}
 (vertex equality $v = w$, vertex membership $v \in V$, edge membership $e \in
 E$ and vertex-edge incidence $\Inc(v, e)$) via negation $\neg$, conjunction
 $\wedge$, disjunction $\vee$, implication $\to$ and equivalence 
 $\leftrightarrow$ together with existential quantification $\exists$ and 
 universal quantification $\forall$ over variables in our domain $V \cup E$. 
 To extend this monadic second order logic (MSOL) to \emph{counting}
 monadic second order logic (CMSOL), one additionally allows the use of
 predicates $\bmod_{p, q}(S)$ for sets $S$, which are true, if and only if $|S|
 \bmod q = p$, for constants $p$ and $q$ (with $p < q$). \par
 Let $\phi$ denote a predicate without unquantified (so-called \emph{free})
 variables constructed as explained above and $G$ be a graph. We call $\phi$ a
 \emph{sentence} and denote by $G \models \phi$ that $\phi$ yields a truth
 assignment when evaluated with the graph $G$. 
 \begin{definition}
 	Let $P$ denote a graph property.
 We say that $P$ is
 \emph{(C)MSOL-definable}, if there exists a (C)MSOL-sentence $\phi_P$ such
 that $P(G)$ if and only if $G \models \phi$. 
 \end{definition}
  We distinguish between two types of free variables.
  Consider a predicate $\phi$ with free variables $x_1,\ldots,x_p$. A subset of
  $x_1,\ldots,x_p$, say $x_1,\ldots,x_a$ (where $a \le p$), can be considered
  its \emph{arguments}, and the variables $x_{a+1},\ldots,x_p$ are its
  \emph{parameters}. We denote this predicate as $\phi(x_1,\ldots,x_a)$, i.e.\
  its parameters do not appear in the notation. We illustrate the difference
  between arguments and parameters in the following example.
 \begin{example}
  Let $P$ denote the property that a graph has a $k$-coloring and
  $\phi_{col}(v, w)$ a predicate, which is true, if and only if a vertex $v$
  has a lower numbered color than $w$ in a given coloring. Then $\phi_{col}$ has
  two arguments, vertices $v$ and $w$, and $k$ parameters, the $k$ color classes. Clearly, the
  choice of the parameters influences the evaluation of $\phi_{col}$, but in
  most applications of parameters for predicates, it is sufficient to show
  that one can guess \emph{some} variables of the evaluation graph to define a
  property.
 \end{example}
 Now, let $R(x_1,\ldots,x_r)$ denote a relation with arguments $x_1,\ldots,x_r$.
 We say that $R$ is \emph{(C)MSOL-definable}, if there exists a parameter-free
 predicate $\phi_R(x_1,\ldots,x_r)$, encoding the relation $R$. Furthermore we
 call $R$ \emph{existentially (CMSOL)-definable}, if there exists a predicate
 $\phi_R(x_1,\ldots,x_r)$ with parameters $x_1,\ldots,x_p$, which, after
 substituting the parameters by fixed values in the evaluation graph, encodes
 the relation $R$.
\par
A central concept used in this paper is an implicit representation
of tree decompositions in monadic second order logic, as we cannot refer to
its bags and edges as variables in MSOL directly. We have to define predicates,
which encode the construction of a tree decomposition of each member of a given
graph class. We require two types of predicates. The $\Bag$-predicates will
allow us to verify whether a vertex is contained in some bag and whether any
vertex set in the graph constitutes a bag in its tree decomposition.
Each bag will be associated with either a vertex or an edge
in the underlying graph (its \emph{witness}) together with some \emph{type},
whose definition depends on the graph class under consideration.
The $\Parent$-predicate allows for identifying edges in the tree decomposition,
i.e. for any two vertex sets $S_p$ and $S_c$, this predicate will be true if and
only if both $S_p$ and $S_c$ are bags in the tree decomposition and $S_p$ is the
bag corresponding to the parent node of $S_c$.
\begin{definition}[MSOL-definable tree decomposition]
 	A tree decomposition $(T = (N, F), X)$ of a graph $G = (V, E)$ is called
 	\emph{existentially MSOL-definable}, if the following are
 	existentially MSOL-definable (with parameters $x_1,\ldots,x_p$ for some
 	constant $p$).
 	\begin{enumerate}[label={(\roman*)}]
 	  \item Each bag $X_p, p \in N$ in the tree decomposition is associated with
 	  either a vertex $v \in V$ or an edge $e \in E$ (called its
 	  \emph{witness}) and can be identified by one of the following predicates
 	  (where $S \subseteq V$ and $s$ and $t$ are constants).
 	  \begin{enumerate}[label={(\alph*)}]
 	    \item $\Bag_{\tau_1}(v, S),\ldots,\Bag_{\tau_t}(v, S)$: The vertex set
 	    $S$ forms a bag in the tree decomposition of $G$, i.e.\
 	    $S = X_p$ for some $p \in N$, it is of type $\tau_i$ ($1 \le i \le
 	    t$) and its witness is $v$.
 	    \item $\Bag_{\sigma_1}(e, S),\ldots,\Bag_{\sigma_s}(e, S)$: The vertex set
 	    $S$ forms a bag in the tree decomposition of $G$, i.e.\ $S = X_p$ for some
 	    $p \in N$, it is of type $\sigma_i$ ($1 \le i \le s$) and its
 	    witness is $e$.
 	  \end{enumerate}
 		\item Each edge in $F$ can be identified with a predicate $\Parent(S_p,
 		S_c)$, where $S_p, S_c \subseteq V$: The vertex sets $S_p$ and $S_c$ form
 		bags in $(T, X)$, i.e.\ $S_p = X_p$ and $S_c = X_c$ for some $p, c \in N$,
 		and $p$ is the parent node of $c$ in $T$.
 	\end{enumerate}
\end{definition}
\begin{lemma}\label{lemExDefTD}
	Let $(T, X)$ be an existentially MSOL-definable tree decomposition with
	parameters $x_1,\ldots,x_p$.
	There exists a predicate $\phi$ with zero parameters and $p$ arguments, which
	is true if and only if the predicates $\Bag_{\tau_1},\ldots,\Bag_{\tau_t}$,
	$\Bag_{\sigma_1},\ldots,\Bag_{\sigma_s}$ and $\Parent$ describe a width-$k$
	rooted tree decomposition of an evaluation graph $G$.
\end{lemma}
\begin{proof}
	The proof can be done analogously to the proof of Lemma 4.7 in \cite{Kal00}.
\end{proof}
A fundamental result about definable graph properties, which we use extensively
throughout our proofs, states
that one can define any edge orientation of partial $k$-trees in MSOL. 
For an in-depth study of
MSOL-definable edge orientations on graphs, see \cite{Cou95}. 
\begin{lemma}[Lemma 4.8 in \cite{Kal00}]\label{lemtwkOrd} Any direction over a
subset of the edges of an undirected graph of treewidth at most $k$ is
existentially MSOL-definable with $k+2$ parameters.
\end{lemma}
The idea of the proof of Lemma \ref{lemtwkOrd} is to find a $(k+1)$-coloring
$\gamma : V \to \{1,\ldots,k+1\}$ (expressed in MSOL by $k+1$ vertex sets) of
the graph and an edge set $F$, such that an edge $e = \{v, w\}$ is directed from
$v$ to $w$, if and only if $\gamma(v) < \gamma(w)$ and $e \in F$ or $\gamma(v) >
\gamma(w)$ and $e \notin F$. Hence, by the choice of the set $F$ we can define
any orientation on the edges of a graph in MSOL, if some $(k+1)$-vertex coloring
of the graph can be fixed.

\subsection{Tree Automata for Graphs of Bounded Treewidth}\label{secPrelTA}
We briefly review the concept of tree automata and recognizability of graph
properties for graphs of bounded treewidth. For an introduction to the topic we
refer to \cite[Chapter 12]{DF13}. For the formal details of the following
notions, the reader is referred to \cite{Kal00}. \par
A tree automaton $\cA$
is a finite state machine accepting as an input a tree structure over an 
alphabet $\Sigma$ as opposed to words in classical word automata.
Formally, $\cA$ is a triple $(\cQ, \cQ_{Acc}, f)$ of a set of states $\cQ$, a
set of accepting states $\cQ_{Acc} \subseteq \cQ$ and a transition function $f$,
deriving the state of a node in the input tree $\cT$ from the states of its children and its
own symbol $s \in \Sigma$.
$\cT$ is \emph{accepted} by $\cA$, if the state of the root node of $\cT$
is an element of the accepting states $\cQ_{Acc}$ (after a run of $\cA$ with
$\cT$ as an input).
\par 
To recognize a graph property on graphs of treewidth at most $k$, one encodes a
rooted width-$k$ tree decompositions as a labeled tree over a special type of
alphabet, in the following denoted by $\Sigma_k$ (see Definition 3.5,
Proposition 3.6 in \cite{Kal00}). We say that a tree
automaton over such an alphabet \emph{processes} width-$k$ tree decompositions.
\begin{definition}[Recognizable Graph Properties]
	Let $P$ denote a graph property. 
	We call $P$ \emph{recognizable} (for graphs of treewidth $k$), if there
	exists a tree automaton $\cA_P$ processing width-$k$ tree decompositions, such
	that following are equivalent.
	\begin{enumerate}[label={(\roman*)}]
	  \item $(T, X)$ is a width-$k$ tree decomposition of a graph $G$ with $P(G)$.
	  \item $\cA_P$ accepts (the labeled tree over $\Sigma_k$ corresponding to)
	  $(T, X)$.
	\end{enumerate}
\end{definition}
Kaller has shown that Courcelle's Conjecture follows immediately from the
construction of an MSOL-definable tree decomposition.
\begin{lemma}[Lemma 5.4 in \cite{Kal00}]\label{lem5_4_Kal}
	Let $P$ denote a graph property, which is recognizable for graphs of bounded
	treewidth. Suppose that there is an MSOL-definable tree decomposition of width
	at most $k$ for any partial $k$-tree $G$. Then, one can write a CMSOL-sentence
	$\Phi$, such that $G \models \Phi$ if and only if $P(G)$.
\end{lemma}

\section{The Main Result}\label{secKOPG}
In this section we investigate Courcelle's Conjecture in the context of
$k$-outerplanar graphs (see Definition \ref{defKOPG}). 
Bodlaender has shown that every $k$-outerplanar graph has treewidth at most
$3k-1$ \cite[Theorem 83]{Bod98}, using the following properties of maximal
spanning forests of a graph.
\begin{definition}[Vertex and Edge Remember Number]
	Let $G = (V, E)$ be a graph with maximal spanning forest $T = (V, F)$. The
	\emph{vertex remember number} of $G$ (with respect to $T$), denoted by
	$vr(G, T)$, is the maximum number over all vertices $v \in V$ of fundamental
	cycles (in $G$ given $T$) that use $v$. Analogously, we define the \emph{edge
	remember number}, denoted by $er(G, T)$.
\end{definition}
In particular, Bodlaender gave a constructive proof that the treewidth of a
graph is bounded by at most $\max\{vr(G, T), er(G, T) + 1\}$ \cite[Theorem
71]{Bod98}. The idea of the proof
is to create a bag for each vertex and edge in the spanning tree, containing the
vertex itself (or the two endpoints of the edge, respectively) and one endpoint
of each edge, whose fundamental cycle uses the corresponding vertex/edge. 
The tree structure of the decomposition is inherited by the structure of the 
spanning tree. 
He then showed, that in a $k$-outerplanar graph $G$ one can
split the vertices of degree $d > 3$ into a path of $d-2$ vertices of degree
three without increasing the outerplanarity index of $G$ (the so-called
\emph{vertex expansion step}, see Figure \ref{figExpStepEx}). In this expanded graph
$G'$ one can find a spanning tree of vertex remember number at most $3k-1$ and edge remember number at most
$2k$ \cite[Lemmas 81 and 82]{Bod98}. Using \cite[Theorem 71]{Bod98}, this yields
a tree decomposition of width at most $3k-1$ for $G'$ and by
simple replacements one finds a tree decomposition for $G$ of the same width.
A constructive version of this proof was given by Katsikarelis \cite{Kat13}.
The expansion step is the major challenge in defining a tree decomposition of a
$k$-outerplanar graph in monadic second order logic, since we cannot use these
newly created vertices as variables. We find an implicit representation of this
step in Section \ref{secKOPGVertExp}. We show how to construct an existentially MSOL-definable
tree decomposition of a 3-connected $k$-outerplanar graph in Section \ref{secKOP3C}
and for the general case of $k$-outerplanar graphs in Section \ref{secG3CC}.
\par

\subsection{An Implicit Representation of the Vertex Expansion
Step}\label{secKOPGVertExp}
\begin{figure}
	\centering
	\subfloat[before expansion]{
		\includegraphics[width=.31\textwidth]{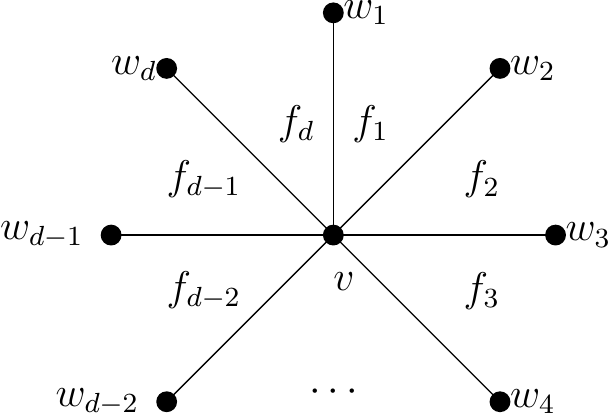}
		\label{figExpStepEx1}
	}
	\qquad
	\subfloat[after expansion]{
		\includegraphics[width=.565\textwidth]{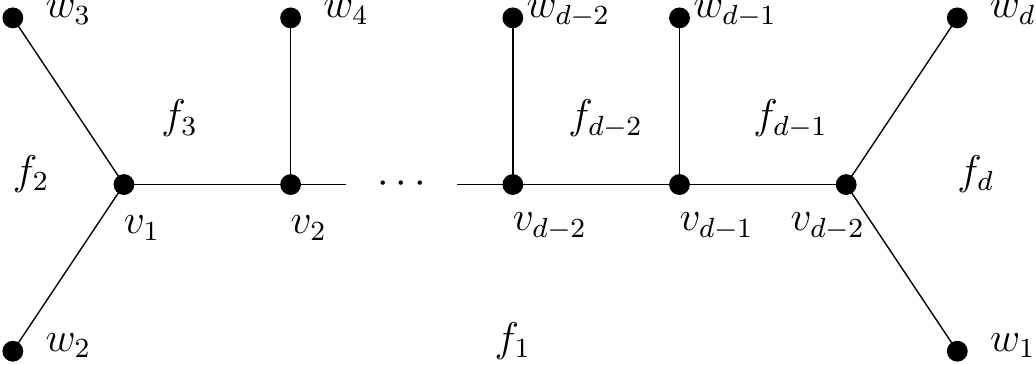}
		\label{figExpStepEx2}
	}
	\caption{Expanding a vertex $v$, where $f_1$ is a layer with lowest layer
 	number.}
 	\label{figExpStepEx}
\end{figure} 
As outlined before, the central step in constructing a width-$(3k-1)$ tree
decomposition of a $k$-outerplanar graph $G$ is splitting the vertices of degree
$d > 3$ into a path of $d-2$ vertices of degree $3$ without increasing the
outerplanarity index of the graph $G$ (see above).
Since we cannot mimic this expansion
step in MSOL directly, we have to find another characterization of this
method, the first step of which is to partition the vertices of a
$k$-outerplanar graph into its \emph{stripping layers}.
\begin{definition}[Stripping Layer of a $k$-Outerplanar Graph]
	Let $G$ be a $k$-outerplanar graph. Removing the vertices on the outer face of
	an embedding of $G$ is called a \emph{stripping step}. When applied
	repeatedly, the set of vertices being removed in the $i$-th stripping step is
	called the $i$-th \emph{stripping layer} of $G$, where $1 \le i \le k$.
\end{definition}
\begin{lemma}\label{lemKOPGLayerDef}
	Let $G = (V, E)$ be a $k$-outerplanar graph. The partition of $V$ into the
	stripping layers of $G$ is existentially MSOL-definable with $k$ parameters.
\end{lemma}
\begin{proof}
We first introduce another characterization of stripping layers of
$k$-outerplanar graphs, which we can use later to define our predicates.
\begin{proposition}\label{propKOPGL}
	Let $G = (V, E)$ be a $k$-outerplanar graph. A partition $V_1,\ldots,V_k$ of
	$V$ represents its stripping layers, if and only if:
	\begin{enumerate}[label={(\roman*)}]
	  \item $G[V_i]$ is an outerplanar graph for all
	  $i = 1,\ldots,k$.\label{propKOPGL1}
	  \item For each vertex $v \in V_i$, all its adjacent vertices are
	  contained in either $V_{i-1}, V_i$ or $V_{i+1}$.\label{propKOPGL2}
	\end{enumerate}
\end{proposition}
\begin{proof}
	($\Rightarrow$) Since in each step we remove the vertices on the outer face of
	the graph, it is easy to see that \ref{propKOPGL1} holds.
	For \ref{propKOPGL2},
 	suppose not. Wlog.\ assume that $v \in V_i$ has a neighbor $w$ in $V_{i+2}$.
 	Before stripping step $i$, $v$ lies on the outer face. Now, for $w$ to not lie
 	on the outer face after stripping step $i$, there needs to be a cycle crossing
 	the edge $\{v, w\}$, hence the embedding of $G$ is not planar and we
 	have a contradiction. \par
 	($\Leftarrow$) We use induction on $k$. The case $k = 1$ is trivial. Now
 	assume that $G = (V, E)$ is an $\ell$-outerplanar graph with a
 	partition of $V$ into $V_1,\ldots,V_\ell$ such that our claim holds. Let
 	$V_{\ell+1}$ be a set of vertices with neighbors only in $V_{\ell+1}$ and
 	$V_\ell$. We denote the corresponding edge set by $E_{\ell+1}$. Clearly,
 	placing the vertices in $V_\ell$ on the outer face results in an $(\ell + 1)$-outerplanar
 	embedding of the graph $G' = (V \cup V_{\ell + 1}, E \cup E_{\ell+1})$.
 	However, some vertices in $V_\ell$ might still lie on the outer face. Denote
 	this vertex set by $V_\ell^O$. We let $V_{\ell + 1}' = V_{\ell + 1} \cup
 	V_\ell^O$ and $V_\ell' = V_\ell \setminus V_\ell^O$. Then, the partition
 	$V_1,\ldots,V_{\ell-1}, V_{\ell}', V_{\ell+1}'$ satisfies our claim and the
 	result follows (reversing the indices of the sets in the partition).
 	 \end{proof}
	It is well known that a graph is outerplanar if it does not contain
	$K_4$, the clique of four vertices, and $K_{2, 3}$, the complete bipartite
	graph on two and three vertices, as a minor (cf.\ \cite[p. 112]{Die12},
	\cite{Sys79}). 
	Borie et al. showed that the fixed minor relation is MSOL-definable \cite[Theorem
	4]{BPT92}, so in our definition we use the predicates $\Minor_{K_4}$ and
	$\Minor_{K_{2, 3}}$ for stating the respective minor containment. The rest can
	be done in a straightforward way according to Proposition \ref{propKOPGL}. The
	details of the predicates can be found in Appendix \ref{appSecMSOLKOPG}, which
	conclude the proof of Lemma \ref{lemKOPGLayerDef}.
 \end{proof}
\begin{definition}[Layer Number]
	Let $G = (V, E)$ be a planar graph. The \emph{layer number} of a face is
	defined in the following way. The outer face gets layer number 0. Then, for
	each other face, we let the layer number be one higher than the minimum layer
	number of all its adjacent faces.\footnote{Unless stated otherwise, we call to
	faces adjacent, if they share an incident vertex.}
\end{definition}
\begin{proposition}\label{propKOPSLLN}
	Let $G = (V, E)$ be a $k$-outerplanar graph, $V_1,\ldots,V_k$ its stripping
	layers and $v \in V_i$. Each face $f$ incident to $v$ has either layer number
	$i$ or $i-1$. Furthermore, $f$ has layer number $i-1$,
	if the boundary of $f$ contains a vertex $w$ with $w \in V_{i-1}$.
\end{proposition}
\begin{proof}
	We observe that removing all vertices on the outer face makes a face of layer
	number $i$ become a face of layer number $i-1$ and our claim follows.
 \end{proof}
The expansion step does not preserve facial adjacency, so in order to not
increase the outerplanarity index of the graph, one makes sure that all faces
are adjacent to a face with lowest layer number. We illustrate the expansion
step of a vertex in Figure \ref{figExpStepEx}.
Following the ideas of the proofs given in \cite[Section 13]{Bod98}, we define
another type of \emph{remember number} to implicitly represent the expansion
step for creating a tree decomposition of a $k$-outerplanar graph.
\begin{definition}[Face Remember Number]
	Let $G = (V, E)$ be a planar graph with a given embedding $\cE$ and $T = (V,
	F)$ a maximal spanning forest of $G$. The \emph{face remember number} 
	of $G$ w.r.t.\ $T$, denoted by $fr(G,T)$ 
	is the maximum number of fundamental cycles $C$ of $G$ given $T$, 
	such that $bd_E(f) \cap E(C) \neq \emptyset$, where
	$bd_E(f)$ denotes the boundary edges of a face $f$, over all 
	faces $f$ in $\cE$, excluding the outer face.
\end{definition}
\begin{figure}
	\centering
	\includegraphics[width=.3\textwidth]{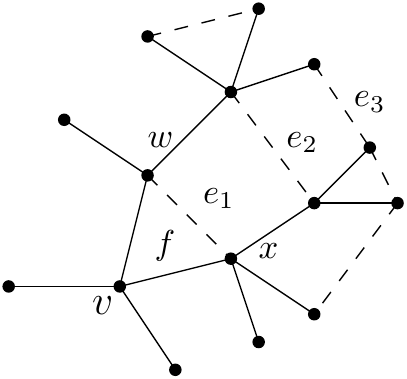}
	\caption{A spanning tree of a planar graph with some additional edges (dashed
	lines).
	The remember number of the face $f$, bounded by $bd(f) = \{v, w,
	x\}$, is 3 in this graph, since the fundamental cycles of the edges $e_1$,
	$e_2$ and $e_3$ intersect with $bd_E(f)$.}
	\label{figFREx}
\end{figure}
For an illustration of face remember numbers, see Figure \ref{figFREx}.
Now, consider the vertex $v_1$ in Figure \ref{figExpStepEx2} and let $e$ 
be an edge whose fundamental cycle $C_e$ uses $v_1$ in some spanning tree of
$G'$.
We observe that $C_e$ intersects with one of the face boundaries of $f_1$, $f_2$
or $f_3$.
Since $v_1$ is a vertex in the expanded graph, we know that in each tree
decomposition based on a spanning tree of $G'$ there will be a bag containing
one endpoint of each edge, whose fundamental cycle intersects with the face
boundary of $f_1$, $f_2$ or $f_3$.
Using this observation, we can also show that one can find a tree decomposition
of a planar graph, whose width is bounded by the face remember
number of a maximal spanning forest, without explicitly expanding vertices.
\begin{lemma}\label{lemPlTWERFR}
	Let $G = (V, E)$ be a planar graph with maximal spanning forest $T =
	(V, F)$. The treewidth of $G$ is at most $\max\{er(G, T) + 1, 3 \cdot fr(G,
	T)\}$.
\end{lemma}
\begin{proof}
	Recall the vertex expansion step and see Figure \ref{figExpStepEx} for an
	illustration. In the following, we will construct a tree decomposition $(T,
	X)$ of the unexpanded graph $G$, imitating the ideas of the expansion step.
	That is, for each vertex $v \in V$ we create a path in $(T, X)$
	in the following way. First, we add $v$ to each of these bags. Let $f_1$ denote
	a face with lowest layer number of all faces incident to $v$ and let all face
	indices be as depicted in Figure \ref{figExpStepEx1}.\footnote{Note that by
	by Proposition \ref{propKOPSLLN}, this number will
	be either $i$ or $i-1$, if $v \in V_i$.}	
	Let $C(f_i)$ denote the set, containing one endpoint of each edge $e \in E
	\setminus F$, whose fundamental cycle $C_e$ intersects with the edge set of the
	boundary of the face $f_i$, i.e.\ $bd_E(f_i) \cap E(C_e) \neq \emptyset$. 
	Let $\deg(v) = d$. We create bags containing the vertices in
	$C(f_1) \cup C(f_i) \cup C(f_{i+1})$, where $i = 2,\ldots,d-1$. (For an edge
	$e_i$ incident to $v$, $f_i$ and $f_{i+1}$ are its incident faces.)
	We make two bags adjacent, if they share
 	two sets $C(f_i)$ and $C(f_j)$ and belong to the same vertex. Note that this
 	way we precisely imitate the construction of bags for the artificially created
 	vertices during the expansion step.
	\par
	Furthermore, for each edge $e_i \in F$, we create a bag containing both its
	endpoints and one endpoint of each edge $e_{fc} \in E \setminus F$, whose
	fundamental cycle uses $e$.
	We observe that the set $C(f_i) \cup C(f_j)$ contains precisely one
	vertex for each such edge $e_{fc}$, where $f_i$ and $f_j$ are the two faces
	incident to $e_i$.
	We then make this bag adjacent to each bag created in the step before, which 
	corresponds to both $C(f_i)$ and $C(f_j)$ and one more set
	$C(f')$. For each incident vertex there will always be precisely one such bag
	and hence, each edge bag will have two neighbors in the tree
	decomposition (one for each endpoint). For an illustration of the constructed
	part of the tree decomposition, see Figure \ref{fig2ConnPTDC}.
	\begin{figure}
		\centering
		\includegraphics[width=.95\textwidth]{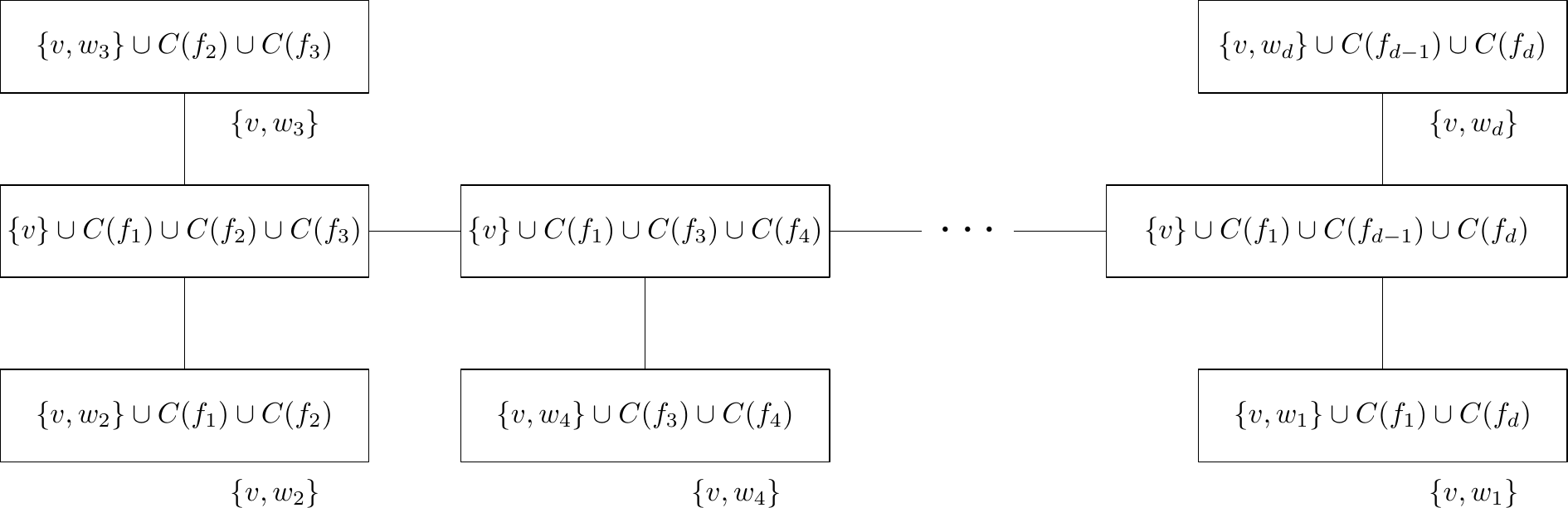}
		\caption{A part of a tree decomposition corresponding to a vertex, as
		used in the proof of Lemma \ref{lemPlTWERFR} (assuming, for explanatory
		purposes, that all incident edges of $v$ are contained in the maximal spanning
		forest of the graph).}
		\label{fig2ConnPTDC}
	\end{figure}
	\par
	One can verify that this construction yields a tree
	decomposition of $G$, and since we know that by definition $|C(f)| \le fr(G,
	T)$ for all faces $f$ (except the outer face) we know that its width is
	bounded by $\max\{er(G, T) + 1, 3 \cdot fr(G, T)\}$.
 \end{proof}
To apply this result to a $k$-outerplanar graph $G$, we show that we can
find a maximal spanning forest of $G$ of bounded edge and face remember number.
\begin{lemma}\label{lemKOPGERFR}
	Let $G = (V, E)$ be a $k$-outerplanar graph. There exists a maximal spanning
	forest $T = (V, F)$ of $G$ with $er(G, T) \le 2k$ and $fr(G, T) \le k$.
\end{lemma}
\begin{proof}
	The proof can be done analogously to the proof of Lemma 81 in \cite{Bod98}.
 \end{proof}

\subsection{3-Connected $k$-Outerplanar Graphs}\label{secKOP3C}
We now show that the construction of the tree decomposition given in the
proofs of Lemmas \ref{lemPlTWERFR} and \ref{lemKOPGERFR} is existentially
MSOL-definable for 3-connected $k$-outerplanar graphs. In Particular we will
make use of the fact that the face boundaries of a 3-connected planar 
graph can be defined by a predicate in monadic second order
logic. We will then define an ordering of all incident edges of a vertex to
create a path in the tree decomposition as described in the proof of Lemma
\ref{lemPlTWERFR}.
\par
A classic result by Whitney states that every 3-connected planar graph has a
unique embedding \cite{Whi32} (up to the choice of the outer face).
Reconstructing this proof, Diestel has shown that the face boundaries of this
embedding can be characterized in strictly combinatorial terms.
\begin{proposition}[Proposition 4.2.7 in \cite{Die12}]\label{prop3ConnFB}
	The face boundaries in a 3-connected planar graph are precisely its
	non-separating induced cycles.
\end{proposition}
We immediately have the following.
\begin{proposition}\label{propKOFB3Conn}
	The face boundaries of a 3-connected planar graph are MSOL-definable.
\end{proposition}
\begin{proof}
	We use Proposition \ref{prop3ConnFB} and define a predicate, which is true if
	and only if a vertex set $V'$ is the face boundary of a 3-connected planar
	graph in the following straightforward way.
	\begin{align*}
		\FaceB_3(V') \Leftrightarrow \Cycle(V', \IncE(V')) \wedge \Conn(V \setminus
		V', E \setminus \IncE(V')) 	
	\end{align*}
	We can use this predicate to define this notion in terms of
	edge sets as well.
	\begin{align*}
		\FaceB_3(E') \Leftrightarrow \FaceB_3(\IncV(E'))
	\end{align*} 
 \end{proof}
Using these observations, we can define predicates encoding the above mentioned
ordering on the incident edges of each vertex. We first need another definition.
\begin{definition}[Face-Adjacency of Edges]
	Let $G = (V, E)$ be a planar graph and $v \in V$. We call two incident edges
	$e, f \in E$ of $v$ \emph{face-adjacent}, if there is a face-boundary
	containing both $e$ and $f$.
\end{definition}
\begin{lemma}\label{lemKOPG3COrd}
	Let $G = (V, E)$ be a 3-connected $k$-outerplanar graph, $v \in V$ with
	$\deg(v) > 3$ and $e_\cA$ an incident edge of $v$, called its \emph{anchor}.
	There exists an ordering $\oriNB(e, f)$, which mimics a clockwise (or
	counter-clockwise) traversal (in the unique embedding of $G$) on all incident
	edges of $v$, starting at $e_\cA$, which is existentially MSOL-definable with
	two parameters $e_\cA$ and $e_\cA'$.
\end{lemma}
\begin{figure}[t]
	\centering
	\subfloat[A vertex $v$ with the anchor edge $e_\cA$ and edge $e_\cA'$.]{
		\includegraphics[width=.225\textwidth]{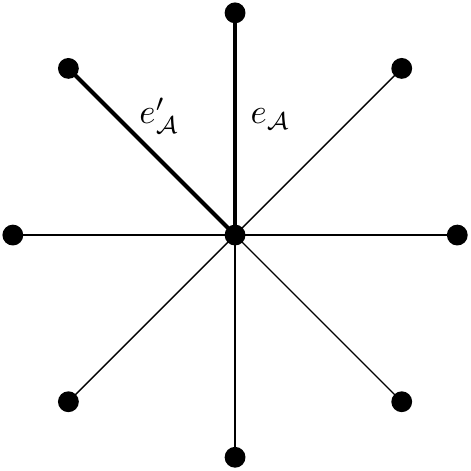}
		\label{figAnchorEx1}}
	\qquad
	\subfloat[The path from $e_\cA$ to $e_i$, according to face adjacency.]{
		\includegraphics[width=.225\textwidth]{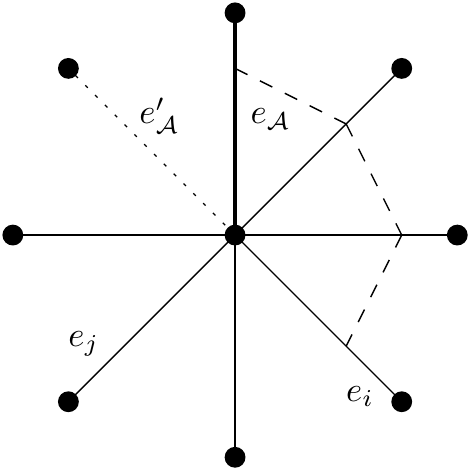}
		\label{figAnchorEx2}
	}
	\qquad
	\subfloat[The path from $e_\cA$ to $e_j$, according to face adjacency.]{
		\includegraphics[width=.225\textwidth]{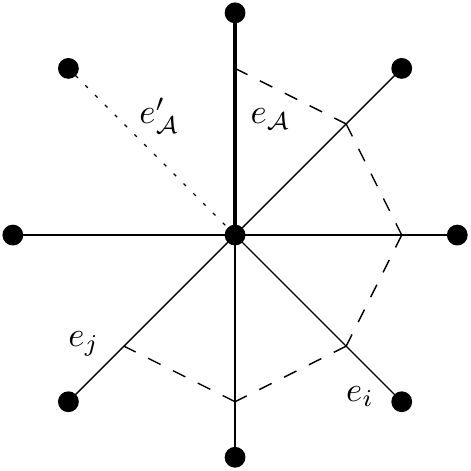}
		\label{figAnchorEx3}
	}
	\caption{A vertex $v$ with two edges $e_i$ and $e_j$, such that $\oriNB(e_i,
	e_j)$ as described in the proof of Lemma \ref{lemKOPG3COrd}, defining a
	clockwise ordering on the incident edges of $v$. Note that paths in the other
	direction starting at $e_\cA$ do not exists, since $e_\cA'$ cannot be included
	in such a path.}
	\label{figAnchorEx}
\end{figure}
\begin{proof}
	We first observe an important property of 2-connected planar graphs, which we
	will use to define the ordering later in the proof.
	\begin{proposition}\label{prop2connPlPDF}
		Let $G = (V, E)$ be a 2-connected planar graph and $v \in V$. Then, all faces
		incident to $v$ are pairwise different.
	\end{proposition}
	\begin{proof}
		Suppose not. Then $\{v\}$ is a separator of $G$.
	 \end{proof}
	Let $e_\cA'$ be another incident edge of $v$, which is also face-adjacent to
	$e_\cA$. (Note that there are exactly two such edges in $G$, the choice of
	which decides whether the ordering is clockwise or counter-clockwise.)
	For any pair of incident edges of $v$, $e_i$ and $e_j$, we let $\oriNB(e_i,
	e_j)$, if and only if we can find sets of edges $E_i$ and $E_j$ with the
	following properties. Let $\Inc(v)$ denote the set of incident edges of $v$.
	\begin{enumerate}[label={(\roman*)}]
	  \item For $\ell = i, j$, the set $E_\ell$ consists of the edge $e_\ell$,
	  $e_\cA$ and a subset of $\Inc(v) \setminus \{e_\cA'\}$ and contains precisely
	  all pairs of face-adjacent edges that, according to face-adjacency, form a
	  path from $e_\cA$ to $e_\ell$.
	  \item $E_i \subset E_j$.
	\end{enumerate}
	For an illustration of the meaning of these edge sets see Figure
	\ref{figAnchorEx}. 
	We now turn to defining this ordering in MSOL. By Proposition
	\ref{prop2connPlPDF}, we know that all faces adjacent to $v$ are pairwise
	different and hence, we can use Proposition \ref{propKOFB3Conn} to define paths
	in terms of face-adjacency in the unique embedding of $G$ between two incident
	edges of $v$. The predicates given in Appendix \ref{appSecMSOL3CKOP1} complete
	the proof.
 \end{proof}
Note that one can lead an alternative proof of Lemma \ref{lemKOPG3COrd}, using
the notion of \emph{rotation systems}, introduced in \cite{Cou00}. Furthermore
one can see that the relation $\oriNB(e, f)$ is existentially MSOL-definable for
a graph $G$ (as opposed to a single vertex, as stated in the Lemma) by replacing
the parameters in the formulation of Lemma \ref{lemKOPG3COrd} with the
corresponding edge set equivalents.

\subsubsection*{Defining the Tree Decomposition}
\begin{lemma}\label{lemKOP3CTDDef}
	Let $G = (V, E)$ be a 3-connected $k$-outerplanar graph. $G$ admits an
	existentially MSOL-definable tree decomposition of width at most $3k$ and
	maximum degree 3 with $4k + 4$ parameters.
\end{lemma}
\begin{proof}
	We mimic the construction given in the proof of Lemma \ref{lemPlTWERFR} and
	use the same notation. 
	We first prove the definability of the spanning tree, upon which the
	construction of our tree decomposition is based.
	\begin{proposition}\label{propKOP3CSpTr}
		Let $G = (V, E)$ be a 3-connected $k$-outerplanar graph. There exists a
		spanning tree $T = (V, F)$ of $G$ with $er \le 2k$ and $fr(G, T) \le k$, which
		is existentially MSOL-definable with one parameter, the edge set $F$ of $T$.
	\end{proposition}
	\begin{proof}
		By Lemma \ref{lemKOPGERFR} we know that such a spanning tree $T$ exists. We
		can use Proposition \ref{propKOFB3Conn} to define $T$ in MSOL, see Appendix
		\ref{appSecMSOL3CKOPGTD}.
	 \end{proof}
	We direct the spanning tree $T$ of Proposition \ref{propKOP3CSpTr} as shown in
	Lemma \ref{lemtwkOrd} to be a rooted tree, using a $3k$-coloring $\Gamma_G$ of
	$G$.
	Note that two colors would already suffice, but we will later use these color
	sets to impose an (arbitrary) orientation on the edges in $E \setminus F$ as well.
	\par
	We now choose the set of anchor and co-anchor edges $E_\cA$ and $E_\cA'$,
	respectively, to fix an ordering on the incident edges of a vertex as shown in
	Lemma \ref{lemKOPG3COrd}.
	For a vertex $v$, let $e_{\ell_1}$ and $e_{\ell_2}$ denote the edges bounding a
	face $f_\ell$ with lowest layer number. 
	(If there is more than one face with lowest
	layer number, we choose the one whose boundary has a shortest face-adjacency
	path from the unique incoming edge in the spanning tree $T$.)
	We then add $e_{\ell_1}$ to $E_\cA$ and
	$e_{\ell_2}$ to $E_\cA'$. Hence, we have that
	$\oriNB(e_{\ell_1}, e)$, for all incident edges $e$ of $v$.
%
%
	\par
	We define three types of bag predicates, all associated with edges. The first
	type, $\sigma$, contains the endpoints of an edge $e \in F$ in the spanning tree of
	$G$ and one endpoint of each edge, whose fundamental cycle uses $e$.
	Note for the following that we can identify an incident face of lowest layer
	number of each vertex by using Proposition \ref{propKOPSLLN} (for details see
	Appendix \ref{appSecMSOL3CKOPGTD}).
	\par 
	We fix an arbitrary orientation on all edges in $E \setminus F$ using the
	coloring $\Gamma_G$ together with the empty edge set (see Lemma
	\ref{lemtwkOrd}).
	Then we define two more types of bags,
	$\sigma_H$ and $\sigma_T$ for each edge $e_i \in \Inc(v) \setminus
	\{e_{\ell_1}, e_{\ell_2}\}$ for all $v \in V$. Let $e_i = \{v, w\}$ with
	orientation from $v$ to $w$, where $f_i$ and $f_{i-1}$ denote the incident
	faces of $e_i$.
	Then, we create a bag of type $\sigma_H$, containing $v$ and one endpoint
	of each edge in $C(v, f_\ell) \cup C(v, f_{i-1}) \cup C(v,
	f_i)$,\footnote{As opposed to the notation in the proof of Lemma
	\ref{lemPlTWERFR}, we use the vertex $v$ as an argument for sets $C$ as well to
	clarify that the faces we are considering in this step are incident faces of
	$v$.} 
	meaning that $\sigma_H$ is a type associated with the head vertex of an edge. 
	We similarly define a type associated with the tail vertex of an edge,
	$\sigma_T$, which is created in the same way as $\sigma_H$, except that it
	contains the tail vertex instead of the head vertex of $e_i$
	(in this case: $w$).
	\begin{figure}[t]
		\centering
		\subfloat[A vertex with incident (directed) edges. Fat edges are in
		the spanning tree.]{
			\includegraphics[width=.325\textwidth]{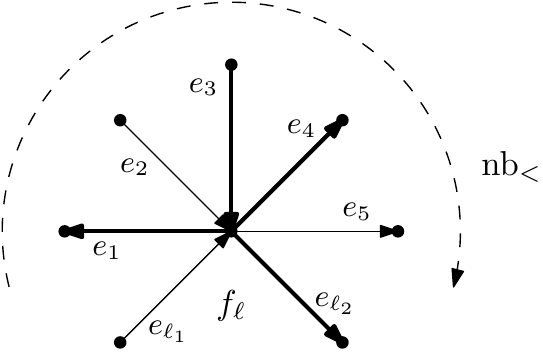}
			\label{figKOP3CTDDef1}
		}
		\qquad
		\subfloat[The corresponding part of the tree decomposition, where the
		edge-orientation describes the $\Parent$-relation.]{
			\includegraphics[width=.57\textwidth]{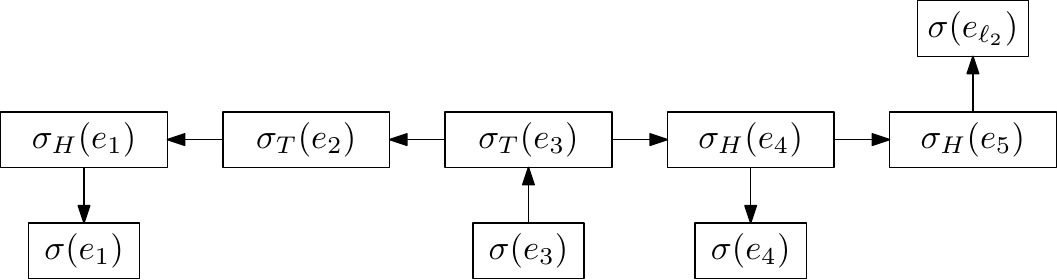}
			\label{figKOP3CTDDef2}
		}
		\caption{A component of a definable tree decomposition as described
		in the proof of Lemma \ref{lemKOP3CTDDef}, corresponding to a vertex $v$ with
		a clockwise ordering on its edges, anchored at $e_{\ell_1}$, where $f_\ell$ is
		a face with lowest layer number of all incident faces of $v$.}
		\label{figKOP3CTDDef}
\end{figure}
	\par
	We now turn to defining the $\Parent$-predicate. For an illustration of any of
	the below mentioned cases, we refer the reader to Figure \ref{figKOP3CTDDef},
	which gives an example of a part of a tree decomposition constructed for a
	vertex.
	\par
	First we consider bags of type
	$\sigma$. Let $e = \{v, w\} \in F$ such that $v$ is its tail vertex and denote
	the corresponding $\sigma$-bag by $X$. Then, we make $X$ the 
	parent of the bag $Y$ of type $\sigma_T$ for the edge $e$. If $v$ is the head
	vertex of $e$, then we make the bag $Y$ of type $\sigma_H$ for the edge
	$e$ the parent of the bag $X$. As mentioned above, we do not create bags of
	type $\sigma_H$ and $\sigma_T$ for the two edges bounding the fixed face with
	lowest layer number $f_\ell$ (for details see the proof of Lemma \ref{lemPlTWERFR}). 
	Let $e_\ell \in \{e_{\ell_1}, e_{\ell_2}\}$. Then, we make
	the bag $X$ of type $\sigma$ corresponding to $e_\ell$ the parent of a bag $Y$
	of type $\sigma_T$ corresponding to an edge $e$, if $e$ and $e_\ell$ bound a
	face together, which is adjacent (in this case, sharing an edge) to the face
	$f_\ell$. Analogously, we make $Y$ the parent of $X$, if $X$ is of type
	$\sigma_H$ for such an edge $e_\ell$.
	%
%
	\par
	Furthermore, we need to add edges between bags of types $\sigma_T$ and
	$\sigma_H$ as well. Note that by now, the only bag, which already has a parent
	is the bag of type $\sigma_T$ for the unique incoming edge $e^* \in F$ in the
	spanning tree of $G$.
	We use the ordering $\oriNB(e, f)$ of the incident edges of a vertex $v$ to
	make sure that the resulting tree decomposition is rooted.
	Let $\oriNBA(e, f)$
	express that two incident edges $e, f$ of $v$ are direct neighbors in the
	ordering $\oriNB(e, f)$. Suppose that $X^*$ is the $\sigma_T$-bag for the edge
	$e^*$ and $Y$ is either a $\sigma_H$- or $\sigma_T$-bag for an edge $f$ with either
	$\oriNBA(e^*, f)$ or $\oriNBA(f, e^*)$. In all of these cases, we make $X^*$
	the parent of $Y$, since $X^*$ already has a parent bag. We observe that we
	have to direct the remaining edges in such a way that they point away from the bag
	$X^*$. Let $e, f \in \Inc(v) \setminus \{e^*, e_{\ell_1}, e_{\ell_2}\}$
	with $\oriNBA(e, f)$, $X$ the $\sigma_H/\sigma_T$-bag of $e$ and $Y$ the
	$\sigma_H/\sigma_T$-bag of $f$. We have to analyze two cases. Note that always
	precisely one of the two holds. 
	\begin{enumerate}[label={(\roman*)}]
	  \item If $\oriNB(e^*, e)$, then make $X$ the parent of $Y$.
	  \item If $\oriNB(f, e^*)$, then make $Y$ the parent of $X$.
	\end{enumerate}
%
	This completes existentially defining the tree decomposition as constructed in
	the proof of Lemma \ref{lemPlTWERFR} in monadic second order logic for a
	3-connected $k$-outerplanar graph.
	\par 
	We now count the parameters used in this proof. To find a face with lowest
	layer number for each vertex, we need the partition into its stripping layers
	as shown in Lemma \ref{lemKOPGLayerDef}. For this step we need $k$ parameters.
	As explained above, for directing the edges of $G$ we use $3k$
	color sets ($G$ has treewidth at most $3k-1$ \cite{Bod98}) and one edge set
	(see Lemma \ref{lemtwkOrd}). We fix edge sets for the spanning tree and the
	anchors $E_\cA$ and co-anchors $E_\cA'$ of the edge ordering $\oriNB(e, f)$.
	Hence, total number of parameters is $4k+4$.
	\par
	The predicates given in Appendix \ref{appSecMSOL3CKOPGTD}
	complete the proof.
\end{proof}

\subsection{Implications of Hierarchical Graph Decompositions to
Courcelle's Conjecture}\label{secG3CC} 
A \emph{block decomposition} of a connected graph $G$ is a tree decompositions,
whose bags contain either the endpoints of a single edge or maximal 2-connected
subgraphs\footnote{Let $G = (V, E)$ be a graph and $W \subseteq V$.
$H = G[W]$ is called a \emph{maximal 2-connected subgraph} of $G$, if $G[W]$ is
2-connected and for all $W' \supset W$, $G[W']$ is not 2-connected.} of $G$
(called the \emph{blocks} of $G$) or a cut-vertex of $G$ (called the
\emph{cuts}) by making a block-bag adjacent to a cut-bag $\{v\}$ if the
block bag contains $v$ (see e.g.\ Section 2.1 in \cite{Die12}).
\par 
Analogously, Tutte showed that given a 2-connected graph (or a
block of a connected graph) one can find a \emph{3-block decomposition} into its
\emph{2-cuts} and \emph{3-blocks}, the latter of which are either 3-connected graphs or
cycles (but not necessarily subgraphs of $G$, see below), which can be
joined in a tree structure in the same way \cite[Chapter 11]{Tut66}
\cite[Section IV.3]{Tut84}.
Courcelle showed that both of these decompositions of a graph
are MSOL-definable \cite{Cou99} and also proved that one can find an
MSOL-definable tree decomposition of width 2, if all 3-blocks of a graph are
cycles \cite[Corollary 4.11]{Cou99}.
In this section, we will use these methods to prove Courcelle's Conjecture for
$k$-outerplanar graphs by showing that the results of the previous section can
be applied to define tree decompositions of 3-connected 3-blocks of a
$k$-outerplanar graph.
\par
As many of our proofs make explicit use of the structure of Tutte's
decomposition of a 2-connected graph into its 3-connected components, we will
now review this concept more closely. 
\begin{definition}[3-Block]
	Let $G = (V, E)$ be a 2-connected graph, $\cS$ a set of 2-cuts of $G$ and
	$W \subseteq V$. A graph $H = (W, F)$ is called a \emph{3-block}, if it can be
	obtained by taking the induced subgraph of $W$ in $G$ and for each incident
	2-cut $S = \{x,y\} \in \cS$, adding the edge $\{x,y\}$ to $F$ (if not already
	present), plus one of the following holds. 
	\begin{enumerate}[label={(\roman*)}]
	  \item $H$ is a cycle of at least three vertices (referred to as a
	  \emph{cycle 3-block}).
	  \item $H$ is a 3-connected graph (referred to as a
	  \emph{3-connected 3-block}).
	\end{enumerate}
\end{definition}
\begin{definition}[Tutte Decomposition]\label{defTutDec}
	Let $G = (V, E)$ be a 2-connected graph. A tree decomposition $(T = (N, F), X)$
	is called a \emph{Tutte decomposition} of $G$, if the following hold. Let $\cS$
	denote a set of 2-cuts of $G$.
	\begin{enumerate}[label={(\roman*)}]
	  \item For each $t \in N$, $X_t$ is either a 2-cut $S \in \cS$ (called the
	  \emph{cut bags}) or the vertex set of a 3-block (called the \emph{block
	  bags}).\label{defTutDec1}
	  \item Each edge $f \in F$ is incident to precisely one cut
	  bag.\label{defTutDec2}
	  \item Each cut bag is adjacent to precisely two block bags.\label{defTutDec4}
	  \item Let $t \in T$ denote a cut node with vertex set $X_t$. Then, $t$ is
	  adjacent to each block node $t'$ with $X_{t} \subset
	  X_{t'}$.\label{defTutDec3}
	\end{enumerate}
\end{definition}
Tutte has shown that additional restrictions can be formulated on the choice of
the set of 2-cuts, such that the resulting decomposition is unique for each
graph (for details see the above mentioned literature). In the following, when we
refer to \emph{the} Tutte decomposition of a graph, we always mean the one that
is unique in this sense, which is also the one that Courcelle defined in his work
\cite{Cou99}. Similarly, by a 3-connected 3-block (cycle 3-block, 2-cut etc.)
of a graph $G$ we mean a 3-connected 3-block in the Tutte decomposition of a block of $G$.
\par
We will now state a property of Tutte decompositions, which will be useful in
later proofs.
\begin{definition}[Adhesion]
	Let $(T = (N, F), X)$ be a tree decomposition. The \emph{adhesion} of $(T, X)$
	is the maximum over all pairs of adjacent nodes $t, t' \in N$ of $|X_t \cap
	X_{t'}|$.
\end{definition}
\begin{proposition}\label{propTutteDecAdh}
	Each Tutte decomposition has adhesion 2.
\end{proposition}
\begin{proof}
	The claim follows directly from Definition \ref{defTutDec} \ref{defTutDec2} and
	\ref{defTutDec3}.
 \end{proof}
 For the proof of the next lemma, we need the notion of \emph{$W$-paths}.
\begin{definition}[$W$-Path]
	Let $G = (V, E)$ be a graph, $W \subseteq V$ and $x, y \in V$. Then, a path
	$P_{xy} = (V_P, E_P)$ between $x$ and $y$ is called a \emph{$W$-path}, if $x, y
	\in W$ and $V_P \cap W = \{x, y\}$, i.e.\ $P_{xy}$ avoids all vertices in $W$
	except its endpoints.
\end{definition}
\begin{lemma}\label{lem3CCPlanar}
	Let $G = (V, E)$ be a 2-connected graph with Tutte decomposition $(T = (N,
	F), X)$. If $G$ is $k$-outerplanar, then all 3-connected 3-blocks
	$C = (W, F)$ of $(T, X)$ are at most $k$-outerplanar.
\end{lemma}
\begin{proof}
	We know that $W = X_t$ for some $t \in N$.
	Let $S = \{x, y\}$ denote a 2-cut of $G$, which is incident to $W$. If $\{x,
	y\} \in E$, we do not have to consider $S$ any further, so in the following, if
	we refer to a 2-cut $S$, we always assume that $\{x, y\} \notin
	E$. Since each such pair $\{x, y\}$ appears in precisely two 3-blocks
	(Definition \ref{defTutDec} \ref{defTutDec4}), we know that there is always at
	least one $W$-path between $x$ and $y$ in $G$.
	\begin{proposition}\label{propTDAdh2SP}
		Let $(T = (N, F), X)$ be a tree decomposition of adhesion 2 and $t \in T$. Let
		$P_1$ and $P_2$ denote two $X_t$-paths. If $P_1$ and $P_2$ share an internal
		vertex, then $P_1$ and $P_2$ have the same endpoints.
	\end{proposition}
	\begin{proof}
		Let $t \in N$. Then, all internal vertices of an
		$X_t$-path $P$ are contained in a set of bags of a unique component $T_t$ of
		$T[N \setminus \{t\}]$. Let $t' \in T_t$ be a neighbor of $t$.
		Then, the endpoints of $P_1$ and $P_2$ are contained in $X_t \cap X_{t'}$.
		Since $(T, X)$ has adhesion 2, both paths have to have the same endpoints.
	 \end{proof} 
	Let $G' = G[W]$ denote the induced
	subgraph of $G$ over the vertex set $W$. For each 2-cut $S$ incident to $W$ we
	add one $W$-path from $G$ to $G$', connecting the two vertices in $S$. Since
	$G$ is planar and $G'$ is a subgraph of $G$, we know that $G'$ is planar. 
	Since $(T, X)$ has adhesion 2 (Proposition \ref{propTutteDecAdh}), we know by
	Proposition \ref{propTDAdh2SP} that there is no pair of $W$-paths
	corresponding to two different incident 2-cuts, sharing an internal
	vertex.
	Hence, we can contract each of these paths
	to a single edge such that the embedding of $G'$ stays planar. Clearly, $G'$
	is isomorphic to $C$ after contraction and the outerplanarity index of $G'$ is
	less than or equal to $k$.
\end{proof}
\begin{figure}[t]
	\centering
	\includegraphics[width=.6\textwidth]{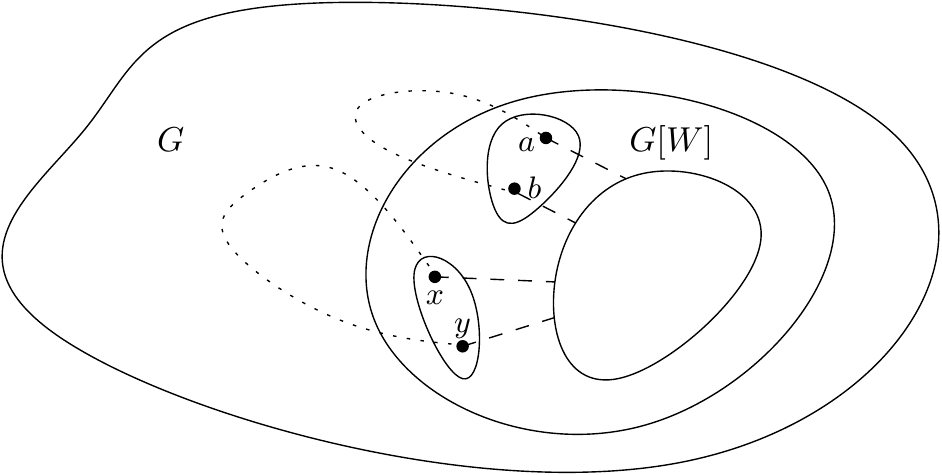}
	\caption{A 2-connected graph $G$ with induced subgraph $G[W]$ over the vertex
	set of a 3-connected 3-block of $G$ with incident 2-cuts $\{a, b\}$ and $\{x,
	y\}$.
	The dashed lines indicate that there might be several edges between a vertex and
	the depicted set and dotted lines represent ($W$-)paths in $G$.}
	\label{fig3ConnGenProof}
\end{figure}
For an illustration of the proof of Lemma \ref{lem3CCPlanar}, see Figure
\ref{fig3ConnGenProof}.
The ideas in this proof can be applied to more
general graph classes as well and we have the following consequence. For the
proof of statement \ref{cor2ConnClasses2}, we need the following definition.
\begin{definition}[Safe Separator \cite{BK06}]
	Let $G = (V, E)$ be a connected graph with separator $S \subset V$. $S$ is
	called a \emph{safe separator}, if the treewidth of $G$ is at most the maximum
	of the treewidth of all connected components $W$ of $G[V \setminus S]$, by
	making $S$ a clique in $G[W]$.
\end{definition}
\begin{corollary}\label{cor2ConnClasses}
	Let $G$ be a 2-connected graph with Tutte decomposition $(T, X)$.
	\begin{enumerate}[label={(\roman*)}]
	  \item If $G$ is planar, then the 3-connected 3-blocks of $(T, X)$
	  are planar.
	  \label{cor2ConnClasses1}
	  \item If $G$ is a partial $k$-tree, then the 3-connected 3-blocks
	  of $(T, X)$ are partial $k$-trees (for $k \ge 2$).\label{cor2ConnClasses2}
	  \item If $G$ is $\cH$-minor free, then the 3-connected 3-blocks of
	  $(T, X)$ are $\cH$-minor free, where $\cH$ is a set of fixed
	  graphs.\label{cor2ConnClasses3}
	\end{enumerate}
\end{corollary}
\begin{proof}
	\ref{cor2ConnClasses1} and \ref{cor2ConnClasses3} follow from the same 
	argumentation (and, clearly, \ref{cor2ConnClasses1} is a consequence of
	\ref{cor2ConnClasses3} by Wagner's Theorem \cite{Wag37}).
	For \ref{cor2ConnClasses2}, we observe the following. By
	\cite[Corollary 4.12]{Cou99} we know that each cut bag $S = \{x, y\}$ is a safe
	separator of $G$ and hence, there is a width-$k$ tree decomposition of $G$
	which has a bag $X_{xy}$ containing both $x$ and $y$. Subsequently, adding the
	edge between $x$ and $y$ does not increase the treewidth of a 3-connected 3-block
	$B_3$.
	(One simply performs a short case analysis of whether $X_{xy}$ is contained in
	the tree decomposition of $B_3$ or not.)
 \end{proof}

\subsubsection*{Replacing Edge Quantification by Vertex Quantification}
As discussed above, a 3-block is in general not a subgraph of a
graph $G$, as we add edges between the 2-cuts of the Tutte
decomposition to turn the 3-blocks into cycles or 3-connected graphs. Since
these absent edges cannot be used as variables in MSOL-predicates (which would
make our logic non-monadic), we need to find another way to quantify over them. \par
In \cite{Cou94}, Courcelle discusses several structures over which one can
define monadic second order logic of graphs, which we will now review.
\begin{definition}[cf.\ Definition 1.7 in \cite{Cou94}]
	Let $G = (V, E)$ be a graph. We associate with $G$ two relational structures,
	denoted by $|G|_1 = \langle V, \edg \rangle$ and $|G|_2 = \langle V \cup E,
	\edg' \rangle$. 
	\begin{enumerate}[label={(\roman*)}]
	  \item All MSOL-sentences and -predicates over $|G|_1$ only use vertices or
	  vertex sets as variables and we have that $\edg(x, y)$ is true for $x, y \in V$, 
	  if and only if there is some edge $\{x, y\} \in E$.
		MSOL-sentences and -predicates over $|G|_2$ use both vertices and edges and
		vertex and edge sets as variables. Furthermore, $\edg'(e, x, y)$ is true if
		and only if $e = \{x, y\}$ and $e \in E$.
		\item If we can express a graph property in the structure $|G|_1$, we call
		it \emph{1-definable} and if we can express a graph property in the
		structure $|G|_2$, we call it \emph{2-definable}.
	\end{enumerate}
\end{definition}
Clearly, the monadic second order logic we are using throughout this paper is
the one represented by the structure $|G|_2$. We use both vertex and edge
quantification and one simply rewrites $\Inc(v, e)$ to $\exists w~\edg'(e, v,
w)$. Since every 1-definable property is trivially also 2-definable, we can
conclude that both 1-definability and 2-definability imply MSOL-definability in
our sense. 
Some of the main results of \cite{Cou94} can be summarized as follows.
\begin{theorem}[\cite{Cou94}]\label{thm1DefEq2Def}
	1-Definability equals 2-definability for
	\begin{enumerate}[label={(\roman*)}]
	  \item planar graphs.
	  \item partial $k$-trees.
	  \item $\cH$-minor free graphs, where $\cH$ is a set of fixed graphs.
	\end{enumerate}
\end{theorem}
Hence, by Theorem \ref{thm1DefEq2Def} we know that we can rewrite each formula
using vertex and edge quantification to one only using vertex quantification,
if a graph is a member of one of these classes.
We will now show that this result can be used to implicitly quantify over
virtual edges of a graph, if these virtual edges can be expressed by an
(existentially) MSOL-definable relation. (For a similar application
of this result, see \cite[Problem 4.10]{Cou99}.)
\begin{lemma}\label{lem2DefVirt}
	Let $G = (V, E)$ be a graph which is a member of a graph class $\cC$ as stated
	in Theorem \ref{thm1DefEq2Def} and let $P$ denote a graph property, which is
	2-definable by a predicate $\phi_P$. Let $E' \subseteq V \times V$ denote a
	set of virtual edges, such that there is a predicate $\edg_{Virt}(v, w)$, which
	is true if and only if $\{v, w\} \in E'$. Then, $P$ is
	1-definable for the graph $G' = (V, E \cup E')$, if $G'$ is a member of $\cC$.
\end{lemma}
\begin{proof}
	By Theorem \ref{thm1DefEq2Def}, $P$ is 1-definable for the graph $G$. Let
	$\phi_{P|1}$ denote the predicate expressing $P$ in $|G|_1$. We replace each
	occurrence of '$\edg(x, y)$' in $\phi_{P|1}$ by '$\edg(x, y) \vee
	\edg_{Virt}(x, y)$' and denote the resulting predicate by $\phi_{P|1}'$, which
	expresses the property $P$ for the graph $G'$ in $|G'|_1$. Since $G' \in \cC$,
	one can replace quantification over sets of virtual edges (or mixed sets of
	edges and virtual edges) by vertex set quantification in the same way as for
	$G$.
\end{proof}
For the specific case of $k$-outerplanar graphs, we can now derive the
following.
\begin{corollary}\label{corKOP2DefVirt}
	Let $G = (V, E)$ be a $k$-outerplanar graph and $P$ a graph property, which is
	(C)MSOL-definable for 3-connected $k$-outerplanar graphs. Let $B_3$ denote a
	3-block of $G$, including the virtual edges between all incident 2-cuts of $B_3$. Then,
	$P$ is (C)MSOL-definable for $B_3$.
\end{corollary}
\begin{proof}
	By \cite[Section 3]{Cou99} we know that there is a predicate $\phi_{\cC_2}(x,
	y)$, which is true, if and only if $\{x, y\}$ is a 2-cut in the Tutte
	decomposition of (a block of) $G$. We know that $B_3$ (including the virtual
	edges) is still $k$-outerplanar (Lemma \ref{lem3CCPlanar}). Hence let
	$\edg_{Virt}(x, y) = \phi_{\cC_2}(x, y)$ and apply Lemma \ref{lem2DefVirt}.
\end{proof}
Note that the statements of Lemma \ref{lem2DefVirt} and Corollary
\ref{corKOP2DefVirt} also hold for existential definability.

\subsubsection*{Defining the Tree Decomposition of a $k$-Outerplanar Graph} 
By Corollary \ref{corKOP2DefVirt} we now know that every graph property, which
can be defined for a 3-connected $k$-outerplanar graph, can also be defined
for a 3-block of any $k$-outerplanar graph $G$ (including its virtual edges).
\par
\begin{figure}[t]
	\centering
	\includegraphics[width=.75\textwidth]{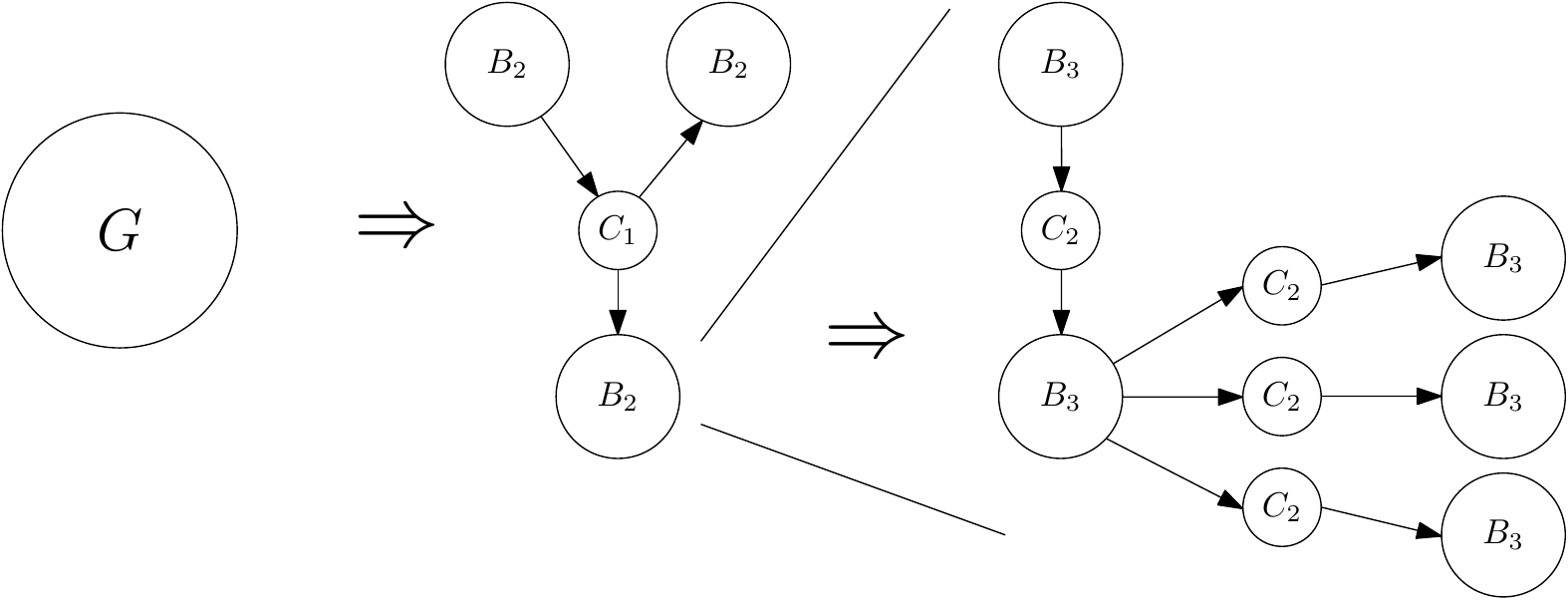}
	\caption{An example hierarchical decomposition of a graph $G$. A bag labeled
	$C_1$ contains a cut-vertex of $G$, $C_2$ a 2-cut of $G$. Bags labeled $B_2$
	contain a 2-block (a single edge or a maximal 2-connected component). If a
	2-block contains a maximal 2-connected component of $G$, it is decomposed
	further into its 2-cuts and 3-blocks, labeled by $B_3$, which contain either a
	cycle or a 3-connected 3-block.}
	\label{fig123ConnDecomp}
\end{figure}
To apply these results to any $k$-outerplanar graph $G$, we first show how to
construct an existentially definable tree decomposition of $G$, assuming that
there exist predicates existentially defining bounded width tree decompositions
for the 3-connected 3-blocks of (the Tutte decomposition of the 2-blocks of) $G$. For an
illustration of the proof idea of the following Lemma, see Figure
\ref{fig123ConnDecomp}, which shows that we can fix a parent-child ordering of
the hierarchical graph decomposition of $G$. After replacing the 3-blocks of $G$
by their corresponding tree decompositions (taking into account the direction
of the edges in the hierarchical decomposition), one can see that we have a
bounded width tree decomposition of the entire graph $G$.
\begin{remark}\label{remSpT}
	Note that in the proofs of the following results, one fixes a root vertex $r
	\in V$ of a $k$-outerplanar graph $G = (V, E)$,
	which will be used to induce a parent-relation on the bags of the hierarchical decomposition of
	$G$ (see Figure \ref{fig123ConnDecomp}). 
	In a later proof, one guesses a rooted spanning tree of $G$, from which one
	derives a set of edges that contains a spanning tree of each 3-connected
	3-block of $G$ (see Lemma \ref{lemSpTrSet}). The root of this spanning tree
	will be precisely this vertex $r$, hence ensuring that we have a conflict-free
	parent-child relation in the resulting tree decomposition of $G$.
\end{remark}
\begin{lemma}\label{lemConnTD3ConnDef}
	Let $G = (V, E)$ be a $k$-outerplanar graph with Tutte decompositions $(T, X)$
	of its 2-connected blocks. Then, $G$ admits an existentially MSOL-definable
	tree decomposition of width at most $3k+3$ with a constant number of
	parameters, if there exist predicates existentially defining width-$3k$ tree
	decompositions for the 3-connected 3-blocks of $G$ with a constant number of
	parameters.
\end{lemma}
\begin{proof}
	Recall the decomposition of a graph into its 3-connected components described
	in the beginning of Section \ref{secG3CC} and see Figure
	\ref{fig123ConnDecomp} for an illustration. We will first show how to
	construct a rooted tree decomposition $(\cT = (\cN, \cF), \cX)$ of $G$ width at
	most $3k + 3$ and then prove that $(\cT, \cX)$ is indeed MSOL-definable.
	Naturally, the description of the tree decomposition is already aimed
	at providing straightforward methods to define its predicates in MSOL.
	\par
	\textbf{I. Constructing the tree decomposition.}
	We use the following notation. $\cC_1$ denotes the set of singletons containing
	a cut-vertex of $G$ and $\cC_2$ denotes the set of 2-cuts in all Tutte
	decompositions of the 2-connected blocks of $G$.
	Furthermore, $\cB_2$ denotes the set of blocks of $G$, $\cB_2^E$ the set of
	blocks that are single edges and $\cB_3$ denotes the set of 3-blocks of $(T,
	X)$. Let $\Theta_{\cB_3} = \{\Theta_1,\ldots,\Theta_r\}$ denote the set of
	tree decompositions of all elements in $\cB_3$. 
	Then, we create a bag in $(\cT, \cX)$ for all elements in $\cC_1$, $\cC_2$,
	$\cB_2^E$ and all bags of each $\Theta_i$ in $\Theta_{\cB_3}$, where $1 \le i
	\le r$. Note that if a 3-block $B_3 \in \cB_3$ is a cycle, one can find a tree
	decomposition of $B_3$ of width 2 directly. We will later study how to
	find an MSOL-definable tree decomposition of such a cycle in a more detailed
	way.
	\par 
	(In the following, keep Remark \ref{remSpT} in mind.)
	We add an edge to $\cF$ between all pairs of adjacent bags originating from a
	tree decomposition $\Theta_i$ with the same orientation.
	To make
	$\cT$ a directed tree, we add edges to $\cF$ between the above mentioned components in
	the following way. First, we fix an arbitrary root $r \in V$ of the graph,
	which is not a member of a cut or a 2-cut of $G$.
	For each vertex $x \in V$, we let $P_x$ denote the paths from $r$ to $x$ in
	$S$ (and sometimes, slightly abusing notation, we might denote it as if it was
	one path, if the meaning of the corresponding statement is clear from the context).
	\par 
	Let $B_2 \in \cB_2 \setminus \cB_2^E$ with Tutte decomposition $(T = (N, F),
	X)$.
	We know that the bags of $(T, X)$ either contain a 2-cut $C_2 \in \cC_2$ or a
	3-block $B_3 \in \cB_3$ with tree decomposition $\Theta_i \in \Theta_{\cB_3}$
	for some $i$ with $1 \le i \le r$.
	We now show which edges we need to add to $\cF$ and how to direct them to
	obtain a rooted tree decomposition of $B_2$ of width at most $3k+2$. We know
	that each edge in $F$ is incident to one cut bag and one block bag (Definition
	\ref{defTutDec}\ref{defTutDec2}, cf.\ Figure \ref{fig123ConnDecomp}). 
	Let $C_2$, $B_3$ and $\Theta_i$ be as above and additionally $C_2 \subset B_3$.
	By Definition \ref{defTutDec}\ref{defTutDec3} we know that there has to be an
	edge in $\cF$ between $C_2$ and one bag in $\Theta_i$, as there is an edge in
	$F$ between $C_2$ and $B_3$ in the Tutte decomposition $(T, X)$.
	We use the following (MSOL-definable) properties to create a rooted tree
	decomposition of a 2-block of $G$.
	\begin{proposition}\label{propC2B3dir}
		Let $C_2 = \{x, y\} \in \cC_2$ and denote by $\cB_3(C_2)$ its (two) neighbors
		in the corresponding Tutte decomposition and $r \in V$ an arbitrarily chosen but
		fixed root vertex, which is not a member of a 2-cut. 
		Then, for each of the following two statements, there is \emph{precisely one}
		3-block $B_3$ which satisfies it.
		\begin{enumerate}[label={(\roman*)}]
		  \item For all $v \in B_3$: $\propSubgraph{P_x}{P_v}$ or
		  $\propSubgraph{P_y}{P_v}$. \label{propC2B3dir1}
		  \item There exists at least one $v \in B_3$, such that
		  $\propSubgraph{P_v}{P_x}$ or $\propSubgraph{P_v}{P_y}$. \label{propC2B3dir2}
		\end{enumerate}
	\end{proposition}
	\begin{proof}
		Observe that $C_2$ separates $G$ into two components, say $G^+$ and $G^-$,
		where $r \in V(G^+)$. Then it immediately follows that \ref{propC2B3dir1}
		holds for the component $B_3^- \in \cB_3(C_2)$ with $B_3^- \subseteq V(G^-)$.
		Now, let $B_3^+ \in \cB_3(C_2) \setminus \{B_3^-\}$. Clearly, $B_3^+
		\subseteq V(G^+)$. Denote by $\cC_2(B_3^+)$ the neighbors of $B_3^+$. Then,
		there is a 2-cut $C_2' \in \cC_2(B_3^+)$, such that \ref{propC2B3dir1} holds
		for $C_2'$ w.r.t.\ $B_3^+$. By definition, we know that there is a vertex $z
		\in C_2 \setminus C_2'$ (where '$\triangledown$' denotes the symmetric
 		difference) 
		and $z$ is also contained in $B_3^+$ (again, by definition).
		Hence, $B_3^+$ satisfies \ref{propC2B3dir2} (with $v = z$).
	\end{proof}
	In case \ref{propC2B3dir1}, we let $C_2 =\{x, y\}$ be the parent bag of $B_3$.
	Recall that $\Theta_i$ denotes a tree decomposition of $B_3$. We add both $x$
	and $y$ to all bags in $\Theta_i$ and make $C_2$ the parent bag of the root of
	$\Theta_i$. \par
	In case \ref{propC2B3dir2}, we let $B_3$ be the parent of $C_2$. Note that
	while a cut bag is always the parent of precisely one block bag, a block bag
	can be the parent of any number of cut bags (cf.\ Figure
	\ref{fig123ConnDecomp}). Hence, adding all vertices of these 2-cuts to the tree
	decomposition $\Theta_i$ could increase the width of $\Theta_i$ to a
	non-constant number. Instead, we observe the following. Since there is a
	(virtual or non-virtual) edge between $x$ and $y$ in $B_3$, we know that there
	is at least one bag containing both $x$ and $y$. Denote the set of such bags by
	$\cX_{xy}$. Since we have to choose precisely one bag in this set to make it a
	parent of $C_2$, we observe the following. Either, there is a bag $\cX^* \in
	\cX_{xy}$, whose parent does \emph{not} contain both $x$ and $y$ or both $x$
	and $y$ are contained in the root bag of $\Theta_i$. In the latter case, we let
	$\cX^*$ be the root of $\Theta_i$. We then make $\cX^*$ the parent of $C_2$.
	\par
	One can verify that this
	yields a rooted tree decomposition of width at most $3k+2$ for any $B_2 \in
	\cB_2 \setminus \cB_2^E$. \par
	To finish the construction of the rooted tree decomposition $(\cT, \cX)$, we
	need to show, which edges to add to $\cF$ between bags in $\cC_1$ and (tree
	decompositions of elements in) $\cB_2$. We use the same idea as before, based
	on a fixed root vertex $r$ in $G$. In the following let $C_1 = \{x\} \in \cC_1$
	and $B_2 \in \cB_2$ with $C_1 \subset B_2$. Since $C_1$ is a separator of $G$,
	one of the following holds for all $v \in B_2$, $v \neq x$.
	\begin{enumerate}[label={(\roman*)}]
	  \item $\propSubgraph{P_x}{P_v}$. \label{lemCTD2CCase1}
	  \item $\propSubgraph{P_v}{P_x}$. \label{lemCTD2CCase2}
	\end{enumerate}
	Again, in case \ref{lemCTD2CCase1}, we make $C_1$ the parent bag of $B_2$. We
	add $x$ to all bags in the tree decomposition of $B_2$ and make $\cX_t$ the
	parent of a bag $\cX_{t'}$, where $\cX_{t'}$ is a bag with
	$\cX_{t'} = B_2$ in case $B_2 \in \cB_2^E$ and if $B_2 \in \cB_2 \setminus
	\cB_2^E$, $X_{t'}$ is the root bag of the tree decomposition of $B_2$,
	constructed as described above. 
	In case \ref{lemCTD2CCase2} we make $B_2$ the parent bag of $C_1$. If $B_2
	\in \cB_2^E$, we simply let the bag $\cX_t$ with $\cX_t = B_2$ be the parent of
	the bag $\cX_{t'}$ with $\cX_{t'} = C_1$. If $B_2 \in \cB_2 \setminus \cB_2^E$, 
	we observe the following. Since $x$ is a cut
	vertex of $G$, no 2-cut of a block of $G$ can contain $x$. Hence we know that
	there exists one unique 3-block $B_3^* \in \cB_3$ with $x \in B_3^*$. We denote
	its tree decomposition by $\Theta_i^*$. Again, we find a bag $\cX_t$ in
	$\Theta_i^*$, such that its parent does not contain $x$. If no such bag exists,
	we let $\cX_t$ be the root of $\Theta_i^*$. We again let $\cX_{t'}$ be the bag
	with $\cX_{t'} = C_1$ and make $\cX_t$ the parent of $\cX_{t'}$. \par
	One can verify that now $(\cT, \cX)$ is a rooted tree decomposition and since
	in the last stage we introduced at most one vertex to each bag of a tree
	decomposition of an element in $\cB_2$, its width is at most $3k+3$. \par
	\textbf{II. Definability.} For defining all necessary predicates for the tree
	decomposition $(\cT, \cX)$, we will refer to $G$ as the graph after adding all
	virtual edges of its Tutte decomposition. We might write down predicates
	quantifying over virtual edges or having virtual edges as free variables, and
	by Corollary \ref{corKOP2DefVirt} we know that all these predicates can be
	defined only using vertex quantification as well.
	\par 
	By some trivial definitions, the statement of
	the lemma, and the results of \cite{Cou99} we know that the predicates listed below
	exist.
	\begin{proposition}[cf.\ \cite{Cou99}]\label{propConnTD3DefUtil}
		Let $G = (V, E)$ be a $k$-outerplanar graph, for whose 2-blocks all
		Tutte decompositions are known. Let $G' = (V, E \cup E')$ denote the graph
		obtained by adding all corresponding virtual edges $E'$ to $G$ and $\gamma : V
		\rightarrow \bN_{|3k+1}$ a coloring of $V$ in $G'$. The following predicates
		are MSOL-definable.
		\begin{enumerate}[label={(\Roman*)}]
		  \item $\Bag_{\cC_1}(v, X)$: $X \in \cC_1$ and $X =
		  \{v\}$.\label{propConnTD3DefUtil1}
		  \item $\Bag_{\cB_2^E}(e, X)$: $X \in \cB_2^E$ and $X = \{v, w\}$, where $e
		  = \{v, w\}$.\label{propConnTD3DefUtil2}
		  \item $\kConn{2}_{\cB_2 \setminus \cB_2^E}(X)$: $X$ is the vertex set of a
		  2-connected 2-block of $G$.\label{propConnTD3DefUtil2a}
		  \item $\Bag_{\cC_2}(v, X)$: $X \in \cC_2$, $v \in X$ and for $w \in X$, $v
		  \neq w$, we have $\gamma(v) < \gamma(w)$.\label{propConnTD3DefUtil3}
		  \item $\kConn{3}_{\cB_3}(X)$: $X$ is the vertex set of a 3-connected 3-block
		  of $G$.\label{propConnTD3DefUtil5}
		  \item $\Cycle_{\cB_3}(X)$: $X$ is a set of vertices forming a cycle block in
		  a 2-block of $G$.\label{propConnTD3DefUtil4}
		  \item $\Bag_{\tau_1}^{\cB_3}(v, X),\ldots,\Bag_{\tau_t}^{\cB_3}(v, X),
		  \Bag_{\sigma_1}^{\cB_3}(e, X),\ldots,\Bag_{\sigma_s}^{\cB_3}(e, X)$: The
		  $\Bag$-predicates of the tree decompositions of the 3-connected 3-blocks
		  of $G$.\label{propConnTD3DefUtil6}
		  \item $\Parent_{\cB_3}(X, Y)$: The $\Parent$-predicate of the tree
		  decompositions of the 3-connected 3-blocks of
		  $G$.\label{propConnTD3DefUtil7}
		\end{enumerate}
	\end{proposition}
	\begin{proof}
		\ref{propConnTD3DefUtil1} and \ref{propConnTD3DefUtil2} follow from
		\cite[Lemma 2.1]{Cou99}, \ref{propConnTD3DefUtil2a} from \cite[Section
		2]{Cou99} and \ref{propConnTD3DefUtil3} from \cite[Section 3]{Cou99} and
		Corollary \ref{corKOP2DefVirt}.
		\ref{propConnTD3DefUtil5} is shown in \cite[Corollary 4.8]{Cou99} and a proof
		of \ref{propConnTD3DefUtil4} can done with the same argument. Finally,
		\ref{propConnTD3DefUtil6} and \ref{propConnTD3DefUtil7} are part of the
		statement of the lemma.
	 \end{proof}
	We now turn to defining tree decompositions for the cycle 3-blocks of a graph,
	after which we only need to show that gluing together all components of our
	construction explained above is MSOL-definable.
	\begin{proposition}\label{propConnTD3DefCycle}
		Let $G = (V, E)$ be a graph and $C = (W, F)$ a cycle 3-block of $G$
		(including virtual edges). There is an existentially definable predicate
		$\Bag_{Cyc}(e, X)$, which is true if and only if $X$ is a bag of a tree
		decomposition of $C$ associated with a (possibly virtual) edge $e$ and an
		existentially definable predicate $\Parent_{Cyc}(X, Y)$ encoding a
		parent-relation of a tree decomposition of $C$.
	\end{proposition}
	\begin{proof}
		Recall that for orienting the edges of our tree decomposition, we first fix a
		root vertex $r$ in the graph $G$ and note that by Proposition
		\ref{propConnTD3DefUtil}\ref{propConnTD3DefUtil5}, $W$ is MSOL-definable. 
		To create a definable tree decomposition of $C$, we now find a root $r_C \in
		W$ of $C$. If $r \in W$, we let $r_C = r$, otherwise we know that there is one
		incident parent cut $C_P \in \cC_1 \cup \cC_2$ of $C$ in $G$. $C_P$ can
		be identified by checking for all 1- and 2-cuts $C_C$, which are incident to $W$,
		if all paths in $S$ from $r$ to the vertices $w \in W$ pass through (at
		least one of the vertices in) $C_C$. This can be defined in a straightforward
		way and one can see that there is always precisely one such cut. If $C_P =
		\{x\} \in \cC_1$, then we let $r_C = x$ and if $C_P = \{x, y\} \in \cC_2$,
		then we let $r_C = x$, if $\gamma(x) < \gamma(y)$ in a fixed coloring $\gamma$
		of $C$. We create a bag $X$ for each edge $f = \{v, w\} \in F$, which is not
		incident to $r_C$ and let $X = \{r_C, v, w\}$. Hence, the predicate $\Bag_{Cyc}(e, X)$
		is also definable in a straightforward way. \par
		We then orient the edges in $F$ in such a way that $C$ is a directed cycle.
		Note that one can find a conflict-free ordering for all cycle blocks in the
		graph $G$. (Otherwise, we might violate the cardinality constraint of MSOL.)
		The predicate $\Parent_{Cyc}(X, Y)$ is true, if and only if the
		following hold.
		\begin{enumerate}[label={(\roman*)}]
		  \item There are two edges $e, f \in F$, such that $\Bag_{Cyc}(e, X)$ and
		  $\Bag_{Cyc}(f, Y)$ (and $e$ and $f$ are contained in the same cycle).
		  \item The directed path from $r_C$ to $\tail(e)$ in $C$ is a strict subpath
		  of the path from $r_C$ to $\tail(f)$.
		  \item $|X \cap Y| = 2$.
		\end{enumerate}
		Note that we only need one additional parameter, the edge set defining the
		edge orientation of $F$, since we already have a coloring for the entire graph
		$G$ (see Proposition \ref{propConnTD3DefUtil}). The details of the predicates
		in Appendix \ref{appSecMSOLHGDCyc} complete the proof.
	 \end{proof}
	To unify the parent-relations for all tree decompositions of 3-blocks, we can
	write
	\begin{equation*}
		\Parent_{\cB_3}'(X, Y) \Leftrightarrow \Parent_{\cB_3}(X, Y) \vee
		\Parent_{Cyc}(X, Y).
	\end{equation*}
	As described above, to create the according parent-relation between blocks
	of the hierarchical decomposition of $G$, we need to add a number of vertices
	to some of the bags of the final tree decomposition $(\cT, \cX)$. The details
	for the changes in those definitions are presented in Appendix
	\ref{appSecMSOLHGDPar}.
	We can define a $\Parent$-predicate for $(\cT, \cX)$ by using the ideas
	explained above to add edges between blocks and cut-bags. Let
	$\Parent_{\cB\cC}(X, Y)$ denote such a predicate. Then, we have that
	\begin{align*}
		\Parent(X, Y) \Leftrightarrow \Parent_{\cB_3}'(X, Y) \vee \Parent_{\cB\cC}(X,
		Y).
	\end{align*}
	To show that the number of parameters that we need to define the above
	mentioned predicates is constant, we note that we only use constructions of
	previous results with constant numbers of parameters. (For the exact number see
	the corresponding result.)
	Note that for the cycle components one additional parameter is as well enough
	(see the proof of Proposition \ref{propConnTD3DefCycle})  to turn
	all cycles into directed cycles, since they are connected in a tree structure 
	in the Tutte decomposition of $G$. Hence, fixing the direction of one
	cycle will always yield the possibility to direct adjacent (i.e.\ sharing a
	2-cut) cycles in a conflict-free manner.
	\par
	The details for the predicate $\Parent_{\cB\cC}(X, Y)$ are
	given in Appendix \ref{appSecMSOLHGDPar} and complete the proof of Lemma
	\ref{lemConnTD3ConnDef}.
\end{proof}
As mentioned in the previous proof, another obstacle in applying Lemma
\ref{lemKOP3CTDDef} to define a tree decomposition for $G$ using its (definable)
hierarchical graph decomposition is the cardinality constraint of MSOL. We
illustrate this problem with an example.
\begin{example}
	Let $G = (V, E)$ be a $k$-outerplanar graph with $\cO(n/\log n)$
	3-connected 3-blocks of size $\cO(\log n)$. Let $P$ denote a graph property,
	which is definable for 3-connected $k$-outerplanar graphs by a predicate
	$\phi_P$. Suppose that $\phi_P$ uses a constant number of parameters. When
	applying $\phi_P$ to all 3-connected 3-blocks of $G$, this might result in a
	predicate using $\cO(n/ \log n)$ parameters and hence, $P$ not definable in this
	straightforward way for $G$.
\end{example}
However, for the case of defining a tree decomposition of a $k$-outerplanar
graph, we can avoid this problem.
When defining a tree decomposition for a
3-connected $k$-outerplanar graph in MSOL, one first guesses a rooted
spanning tree of $G$. 
To avoid guessing a non-constant number of spanning trees,
we will find a set of edges $\cS_E$, which contains a spanning tree with
bounded edge and face remember number for each 3-connected 3-block of $G$.
Furthermore we guess one set $\cR_V$, containing one unique vertex for each
3-connected 3-block of $G$, which we will use as the root of its spanning tree.
We need to make some observations about such candidate sets $\cS_E$ and
$\cR_V$. We first prove the existence of these sets and then their
MSOL-definability.
\begin{lemma}\label{lemSpTrSet}
	Let $G = (V, E)$ be a planar graph and $G = (V, E \cup E')$ the
	graph obtained by adding the virtual edges $E'$ of the Tutte decompositions of
	the 2-connected blocks of $G$ to $G$. Let $T = (V, F)$ be a spanning tree of
	$G$ with $er(G, T) \le \lambda$ and $fr(G, T) \le \mu$. Let $B_3 = (V_{B_3},
	E_{B_3}) \in \cB_3$ be a 3-connected 3-block of $G'$ (including virtual edges)
	and $T_{B_3} = T[V_{B_3}]$.
	One can construct from $T_{B_3}$ a spanning tree $T_{B_3}^*$ of $B_3$ with
	$er(B_3, T_{B_3}^*) \le \lambda$ and $fr(B_3, T_{B_3}^*) \le \mu$ by adding
	edges from $E \cup E'$ to $T_{B_3}$.
\end{lemma}
\begin{figure}
	\centering
	\subfloat[$T_{B_3}$ without edge direction.]{
		\includegraphics[width=.35\textwidth]{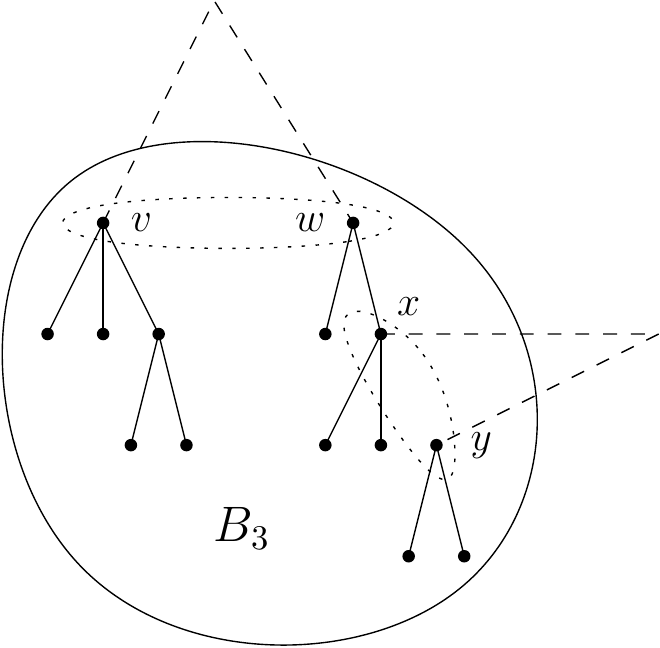}
		\label{figSpTrEx1}
	}
	\qquad
	\subfloat[$T_{B_3}$ with edge direction.]{
		\includegraphics[width=.35\textwidth]{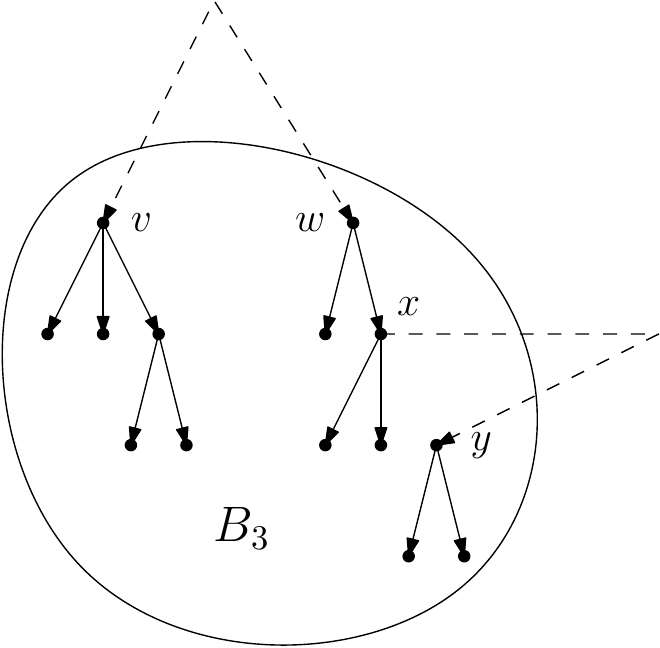}
		\label{figSpTrEx2}
	}
	\caption{A forest $T_{B_3}$ of a 3-connected 3-block of an example graph. The
	dashed lines indicate the paths in $T$ between two endpoints of a incident cut
	of $B_3$. Here, $\{v, w\}$ is the root cut of $B_3$ and $\{x, y\}$ a child
	cut. Note that by Propositions \ref{propSpTrSetCh} and \ref{propSpTrSetPar},
	this small example is already somewhat general.}
	\label{figSpTrEx}
\end{figure}
\begin{proof}
	Clearly, $T_{B_3} = (V_{B_3}, F_{B_3})$ is a forest in $B_3$ and in the
	following we denote its tree components by $F_1 = (V_{F_1}, E_{F_1}),\ldots,F_c
	= (V_{F_c}, E_{F_c})$.
	We will now show how to connect these components to a tree.
	Let $\cC_2' \subseteq \cC_2$ denote the set of incident 2-cuts of $B_3$. 
	\begin{proposition}\label{propB3C2AddConn}
		Let $T_{B_3}'$ denote the graph obtained by adding an edge between all 2-cuts
		$\{x, y\} \in \cC_2'$ in $T_{B_3}$ (if not already present). Then, $T_{B_3}'$
		is connected.
	\end{proposition}
	\begin{proof}
		Let $(T_T = (N_T, F_T), X)$ denote the Tutte decomposition containing $B_3$
		and let $B_3 = X_t$ with $t \in N_T$.
		Let $v, w \in V_{B_3}$ and consider the unique path $P_{vw}$ between $v$ and
		$w$ in $T$. There are two cases: (I) The path $P_{vw}$ is completely contained
		in $B_3$ and $v$ and $w$ belong to the same connected component. (II) Suppose
		that they do not and let $F_i$ denote the component with $v \in V_{F_i}$ and $F_j$ the
		component with $w \in V_{F_j}$. 
		Let $x$ and $y$ be the vertices on the path $P_{vw}$ with $x, y \in V_{B_3}$
		(and $x \neq y$), such that $x$ has a neighbor $x' \notin V_{B_3}$ and $y$ has
		a neighbor $y' \notin V_{B_3}$ (both in $P_{vw}$). Denote this subpath by $P_{xy}$. Then,
		$P_{xy}$ is a $V_{B_3}$-path in $G$. Hence, there is a unique component in
		$T_T' = T_T[N_T \setminus \{t\}]$ containing all internal vertices of
		$P_{xy}$.
		Since the neighbor of $t$ in $T_T'$ is a cut-bag, we know that it has to
		contain both $x$ and $y$ and hence $\{x, y\} \in \cC_2'$.
	\end{proof}
	By Proposition \ref{propB3C2AddConn} we know that we can find a subset of
	incident 2-cuts of each 3-connected 3-block to turn $T_{B_3}$ into a tree.
	We now prove that adding these edges does not increase the edge and face
	remember number. Consider a 2-cut $C_2 = \{x, y\} \in \cC_2'$, such
	that $\{x, y\} \notin F$. Since $T$ is a spanning tree of $G$, we know that
	there is one unique path $P_{xy}$ between $x$ and $y$ in $T$. Let $T_{B_3}' =
	(V_{B_3}', F_{B_3}')$ denote the tree obtained by adding the above described
	paths between the components of $T_{B_3}$. Then, $T_{B_3}'$ is a spanning tree
	of the graph $G_{B_3}' = (V_{B_3}', E_{B_3} \cup F_{B_3}')$ with $er(G_{B_3}',
	T_{B_3}') \le \lambda$ and $fr(G_{B_3}', T_{B_3}') \le \mu$, since $G_{B_3}'
	\sqsubseteq G$ and no edges, which are not members of $T_{B_3}'$, are
	introduced in $G_{B_3}'$.
	Subsequently, replacing each path $P_{xy}$ by a single edge in $T_{B_3}'$ does
	not increase the edge and face remember number as well and after these
	replacements, we have that $T_{B_3}' = T_{B_3}^*$ and our claim follows. For an
	illustration of this proof see Figure \ref{figSpTrEx1}. 
 \end{proof}
\begin{lemma}\label{lemSpTrDirSet}
	The statement of Lemma \ref{lemSpTrSet} also holds, if one replaces the term
	\emph{spanning tree} by \emph{rooted spanning tree}. Furthermore there
	is a set $\cR_V \subseteq V$, which contains precisely one vertex acting
	as a root for a spanning tree for each 3-connected 3-block of $G$.
\end{lemma}
\begin{proof}
	We use the same notation as in the proof of Lemma \ref{lemSpTrSet}. Since $T =
	(V, F)$ is a rooted spanning tree, we know that its components $F_1,\ldots,F_c$
	in $B_3$ are rooted trees as well, see Figure \ref{figSpTrEx2} for an illustration.
	Since the direction between block and cut bags of a Tutte decomposition of a
	block of $G$ are based on the root of the spanning tree $T$ (see Remark
	\ref{remSpT} and the proof of Lemma \ref{lemConnTD3ConnDef}), we observe the following.
	Let $C_2 = \{x, y\} \in \cC_2$ denote an incident 2-cut of $B_3$ with $\{x,
	y\} \notin F$. There are two cases we have to consider. Either, $C_2$ is the
	parent cut of $B_3$ or it is a child cut.
	\begin{proposition}\label{propSpTrSetCh}
		Let $C_2$ be a child cut of $B_3$. Wlog.\ $x$ is a vertex in a tree $F_i$ and
		$y$ is the root of a tree $F_j$.
	\end{proposition}
	\begin{proof}
		Suppose not. We know that there is a path $P_{xy}$ between $x$ and $y$ in $T$.
		If $y$ is a non-root vertex in $F_j$, then we cannot direct the edges of
		$P_{xy}$ in $T$ such that every vertex has precisely one parent. Hence, $T$ is
		not a directed tree and we have a contradiction.
	 \end{proof}
	\begin{proposition}\label{propSpTrSetPar}
		Let $C_2$ be the parent cut of $B_3$. Then, $x$ and $y$ are roots of two trees
		$F_i$ and $F_j$.
	\end{proposition}
	\begin{proof}
		For any vertex $v \in V_{B_3}$ we know by definition (see the proof of Lemma
		\ref{lemConnTD3ConnDef}) that for every vertex $v \in V_{B_3}$, the directed
		path from the root $r$ of $T$ to $v$ in $T$ is either a subpath of the
		directed path from $r$ to $x$ or from $r$ to $y$.
		Hence, neither $x$ nor $y$ can have a parent in $T_{B_3}$.
	 \end{proof}
	We can direct the additional edges using Propositions \ref{propSpTrSetCh}
	and \ref{propSpTrSetPar}. In the case that $C_2$ is a child cut, we can always
	direct the edge $\{x, y\}$ from $x$ to $y$ (using the notation of Proposition
	\ref{propSpTrSetCh}).
	If $C_2$ is the parent cut, we know by Proposition \ref{propSpTrSetPar} that we
	can orient $\{x, y\}$ arbitrarily. There are two cases we need to analyze to
	make sure we do not create a conflicting orientation of $\cS_E$.
	In the first case, the edge $\{x, y\}$ has been added to $\cS_E$ by the parent
	block of $C_2$. We then use the same orientation. In the second case, if $\{x,
	y\} \notin \cS_E$, we can choose the direction arbitrarily.
	\par
	We now turn to finding the set of roots $\cR_V$. If $B_3$ is the root block
	according to the spanning tree of $G$ with root $r_G$, then we add $r_G$ to
	$\cR_V$ as the root of $B_3$. Otherwise, we find its parent cut $C_2 = \{x,
	y\}$. Assume wlog.\ that the edge $\{x, y\}$ is directed from $x$ to $y$
	according to the construction explained above. Then we add $x$ to $\cR_V$.
	Since each cut-bag has precisely one child block bag (Definition
	\ref{defTutDec}\ref{defTutDec2}), we know that this vertex is unique for each
	3-block $B_3$.
 \end{proof}
\begin{lemma}\label{lemSpTrSetsDef}
	The sets $\cS_E$ and $\cR_V$ of Lemmas \ref{lemSpTrSet} and \ref{lemSpTrDirSet}
	are existentially MSOL-definable with $3k+2$ parameters.
\end{lemma}
\begin{proof}
	Let $G = (V, E)$ denote a $k$-outerplanar graph, such that the virtual edges
	introduced by the Tutte decompositions of its 2-connected blocks are already
	included in $E$.
	On a high level, for defining $\cR_V$ and $\cS_E$, we need to encode is the
	following:
	\begin{enumerate}[label={(\roman*)}]
	  \item There are sets $\cR_V \subseteq V$, $F \subseteq E$ and $F' \subseteq
	  E$ with $\cS_E = F \cup F'$.\label{lemSpTrSetsDef1}
	  \item Guess a root $r_T \in V$, such that $F$ is the edge set of a rooted
	  	spanning tree in $G$.\label{lemSpTrSetsDef2}
	  \item An edge $e = \{x, y\}$ is possibly (but not necessarily) a member of
	  	$F'$, if $\{x, y\} \in \cC_2$ and $e \notin F$.\label{lemSpTrSetsDef3}
	  \item For all $B_3 \in \cB_3$, the graph $T_{B_3}^* = (B_3, \cS_E
	  \cap (B_3 \times B_3))$ is a spanning tree of the graph $G_{B_3} = G[B_3]$
	  with $er(G_{B_3}, T_{B_3}^*) \le 2k$ and $fr(G_{B_3}, T_{B_3}^*) \le
	  k$.\label{lemSpTrSetsDef4}
	  \item A vertex $v \in V$ is possibly (but not necessarily) a member of
	  $\cR_V$, if it is a member of a 2-cut $\{v, w\} \in
	  \cC_2$.\label{lemSpTrSetsDef5}
	  \item For each 3-connected 3-block $B_3 \in \cB_3$, there is a vertex
	  $r_{B_3} \in \cR_V$, such that $T_{B_3}^*$ can be rooted at $r_{B_3}$
	  (without altering the edge direction of any other edge in
	  $\cS_E$).\label{lemSpTrSetsDef6}
	\end{enumerate}
	The existence of such sets $\cR_V$ and $\cS_E$ is shown in Lemmas
	\ref{lemSpTrSet} and \ref{lemSpTrDirSet}, so we do not need to encode all
	details mentioned in the corresponding proofs explicitly. Property
	\ref{lemSpTrSetsDef4} is MSOL-definable by Proposition \ref{propKOP3CSpTr},
	since $G_{B_3}$ is 3-connected. \par
	As parameters we have the edge set of the spanning tree and
	again a $3k$-coloring and one edge set to fix the orientation of the edges in
	$\cS_E$.
	\par
	The details of the predicates encoding the rest
	of the properties are given in Appendix \ref{appSecMSOLHGDSE} and complete the
	proof.
 \end{proof}
 We can now use the above results to conclude that we can find predicates
defining tree decompositions of 3-connected 3-blocks of $k$-outerplanar graphs.
\begin{corollary}\label{corKOP3C3BTDDef}
	Let $G = (V, E)$ be a $k$-outerplanar graph. Then, there exist predicates
	existentially defining tree decompositions of width at most $3k$ for each
	3-connected 3-block of $G$ with a constant number of parameters.
\end{corollary}
\begin{proof}
	By Lemma \ref{lemKOP3CTDDef} we know that a 3-connected $k$-outerplanar graph
	admits an MSOL-definable tree decomposition of width $3k$, based on a rooted
	spanning tree of the graph. By Corollary \ref{corKOP2DefVirt} we can define
	such a tree decomposition in a structure, which also includes the virtual edges
	of a 3-block in $G$ (and by Lemma \ref{lem3CCPlanar} we know that this graph
	is still $k$-outerplanar).
	Finally, by Lemmas \ref{lemSpTrSet}, \ref{lemSpTrDirSet} and
	\ref{lemSpTrSetsDef} we know that we can find definable edge and vertex sets
	which contain the edges of spanning trees for each 3-connected 3-block with the
	required bound on their vertex and edge remember numbers without violating the
	cardinality constraint of monadic second order logic. 
	Similarly, we can find sets containing anchor and co-anchor edges for all
	3-connected 3-blocks in a straightforward way. Hence, also for defining the
	ordering of all incident edges of all vertices in a 3-connected 3-block, two
	sets are sufficient.
	Subsequently, the number of parameters involved is
	bounded by a constant. For the exact bounds see the corresponding result.
\end{proof}
Combining Lemma \ref{lemConnTD3ConnDef} and Corollary \ref{corKOP3C3BTDDef}
yields that $k$-outerplanar graphs admit existentially MSOL-definable tree
decompositions of width at most $3k+3$. 
It then follows from Lemma \ref{lem5_4_Kal} that recognizability
implies CMSOL-definability for $k$-outerplanar graphs. In the light of
Courcelle's Theorem \cite{Cou90}, we have the main result of this paper.
\begin{theorem}
	CMSOL-definability equals recognizability for $k$-outerplanar graphs.
\end{theorem}

\section{Conclusion}\label{secConc}
In this paper we have shown that recognizability implies definability in
counting monadic second order logic for $k$-outerplanar graphs, resolving a special
case of a conjecture by Courcelle \cite{Cou90}. Starting at the more
restrictive case of 3-connected $k$-outerplanar graphs, we proved that one can
use hierarchical graph decompositions to define tree decompositions for general
$k$-outerplanar graphs in monadic second order logic. We have also given
indications that this technique might be applicable for other graph classes as
well (see Corollary \ref{cor2ConnClasses}), depending on how their tree
decompositions are defined in MSOL. 
3-Connected graphs often have favorable properties
when it comes to defining graph properties in MSOL. For example, in our proof we
used the fact that the face boundaries of a 3-connected can be expressed in
strictly combinatorial terms and are definable in a straightforward way (see
Propositions \ref{prop3ConnFB} and \ref{propKOFB3Conn}). Hence, we believe that
the techniques presented in this paper can be helpful in resolving the
conjecture in its general statement.

	\bibliography{References}
	
	\newpage
	\appendix
	
	\section{Monadic Second Order Predicates and Sentences}\label{appSecMSOL}
We build sentences in monadic second order logic from a collection
of predicates. Once we defined these predicates they will be the building blocks
of more complex expressions, joined by MSOL-connectives and/or quantification of
its declared variables. Hence, we follow the ideas of the work of Borie et al.
\cite{BPT92}, who also give a large list of predicates and their definitions. \\
Note that the length of our sentences and formulas always has to be bounded by
some constant, independent of the size of the input graph. \par
We will denote single element variables by small letters, where $v, w, v', w',
\ldots$ typically represent vertices and $e, f, e', f',\ldots$ edges. Set
variables will be denoted by capital letters. Unless stated otherwise
explicitly, $V$ always denotes the vertex set of some input graph $G$ and $E$
its edge set. Since we always assume our predicates to appear in the context of
such a graph we might drop these two variables as an argument of a predicate.
\par
By some trivial definition, the following predicates are MSOL-definable (see
also Theorem 1 in \cite{BPT92}). In our text we might refer to them as the
\emph{atomic} predicates of monadic second order logic over graphs.
\begin{enumerate}[label={(\Roman*)}]
  \item $v = w$ (Vertex equality)
  \item $\Inc(e, v)$ (Vertex-edge incidence)
  \item $v \in V$ (Vertex membership)
  \item $e \in E$ (Edge membership)
\end{enumerate}
Note that to shorten our notation we might omit statements such as $v \in V$ or
$e \in E$ when quantifying over a variable. In this case we are referring to some
vertex/edge in the whole graph and the interpretation of the variables will
always be obvious from the context or the notational conventions explained above. \par
From the atomic predicates, one can directly derive the following:
\begin{itemize}
  \item $\Adj(v, w, E)$ (Adjacency of $v$ and $w$ in $E$)
  \item $\Edge(e, v, w)$ ($e = \{v, w\}$)
\end{itemize}
In a straightforward way (and by Theorem 4 in \cite{BPT92}), one can see that
the following are MSOL-definable:
\begin{itemize}
  \item $V = V' \cup V''$, $V = V' \setminus V''$, $V = V' \cap V''$ (plus the
  edge set equivalents)
  \item $V' = \IncV(E')$ [$E' = \IncE(V')$] ($V'$ [$E'$] is the set of incident
  vertices [edges] of $E'$ [$V'$])
  \item $\deg(v, E) = k$ ($v$ has degree $k$ in $E$, where $k$ is a constant)
  \item $\Conn(V, E)$, $\Conn_k(V, E)$, $\Cycle(V, E)$, $\Tree(V, E)$, $\Path(V,
  E)$
  \item $\Minor_H$ (A graph contains a minor $H$ of fixed size)
\end{itemize}

\subsection{Bounded Vertex and Edge Remember Number}\label{appSecMSOLVRER}
In this section we show how to define tree decompositions of graphs for which we
can find a spanning tree with bounded vertex and edge remember
number. Note that this immediately implies a bounded-width tree decomposition
for bounded degree $k$-outerplanar graphs. 
First, we are going to show how to identify an edge set as a spanning tree with
vertex remember number less than or equal to $\kappa$ and edge remember number
less than or equal to $\lambda$, both constant.
\begin{align*}
	\exists E_T &(\Tree(V, E_T) \wedge vr(E_T) \le \kappa \wedge er(E_T) \le
	\lambda) \\
	vr(E_T) \le \kappa \Leftrightarrow &(\forall v \in V)(\forall e_1 \in E
	\setminus E_T)\cdots \forall (e_{\kappa + 1} \in E \setminus E_T)\\
	&\Big(\Big(\bigwedge_{i = 1,\ldots,\kappa + 1} \FundCyc(v, e_i)\Big) \to
	\bigvee_{1 \le i < j \le \kappa + 1} e_i = e_j\Big) \\
	er(E_T) \le \lambda \Leftrightarrow &(\forall e \in E)(\forall e_1 \in E
	\setminus E_T)\cdots \forall (e_{\lambda + 1} \in E \setminus E_T)\\
	&\Big(\Big(\bigwedge_{i = 1,\ldots,\lambda + 1} \FundCyc(e, e_i)\Big) \to
	\bigvee_{1 \le i < j \le \lambda + 1} e_i = e_j\Big)
\end{align*}
In the following, assume that $E_T$ is the edge set of the spanning tree of $G$
(as shown above), which additionally has edge orientations, defined in MSOL by
predicates $\head$ and $\tail$.
\begin{align*}
	\Bag_V(v, X) \Leftrightarrow &v' \in X \leftrightarrow (v' = v \vee (\exists
	e \in E \setminus E_T)(\Inc(v', e) \\ 
	&\wedge \FundCyc(v, e))) \\
	\Bag_E(e, X) \Leftrightarrow &v' \in X \leftrightarrow (\Inc(v', e) \vee
	(\exists e' \in E \setminus E_T)(\Inc(v', e') \\ 
	&\wedge \FundCyc(e, e'))) \\
	\Parent(X_p, X_c) \Leftrightarrow &\exists v(\exists e \in E_T)((\Bag_V(v,
	X_p) \wedge \Bag_E(e, X_c) \wedge \head(v, e)) \\ 
	&\vee (\Bag_V(v, X_c) \wedge \Bag_E(e, X_p) \wedge \tail(v, e))) \\
\end{align*}

\subsection{$k$-Outerplanar Graphs}\label{appSecMSOLKOPG}
Using the forbidden minors ($K_4$ and $K_{2, 3}$), we can define a predicate for
verifying whether a graph is outerplanar in a straightforward way.
\begin{align*}
	\Outerplanar(V', E') \Leftrightarrow \neg (\Minor_{K_4}(V', E') \vee
	\Minor_{K_{2, 3}}(V', E'))
\end{align*}
Following the argumentation in the proof of Lemma \ref{lemKOPGLayerDef}, we can
define our predicate as follows.
\begin{align*}
	\exists V_1 \cdots \exists V_k &\Big(\Part_V(V, V_1,\ldots,V_k) \wedge
	\Outerplanar(V_1, \IncE(V_1)) \\ 
	&\wedge \cdots \wedge \Outerplanar(V_k, \IncE(V_k)) \\
	&\wedge \forall v\Big(\bigwedge_{i = 1,\ldots,k} v \in V_i \to \forall w
	\forall e (\Edge(e, v, w) \\
	&\to (w \in V_{i-1} \vee w \in V_i \vee w \in V_{i + 1})\Big)\Big)
\end{align*}

\subsubsection{3-Connected $k$-Outerplanar Graphs}\label{appSecMSOL3CKOP1}
We first give the necessary definition of defining the ordering $\oriNB$ as
described in Lemma \ref{lemKOPG3COrd}. The first step is to define
face-adjacency of two edges.
\begin{align*}
	\Adj_F(e, f) \Leftrightarrow &\exists v (\Inc(v, e) \wedge \Inc(v, f)) \\
	&\wedge (\exists E' \subseteq E) (\FaceB_3(E') \wedge e \in E' \wedge f \in
	E')
\end{align*}
Next, we define a set to check whether a set of edges is a face-adjacency path
from the one to the other, if they both share a vertex $v$. Intuitively
speaking, this predicate states that each edge in the candidate set $E'$ has
precisely one neighbor in it, if the edge is either $e$ or $f$ and precisely two
otherwise. Furthermore, $E'$ has to consist of a subset of the incident edges of
$v$, without $e_\cA'$ (see the proof of Lemma \ref{lemKOPG3COrd}) and it has to contain
both $e$ and $f$.
\begin{align*}
	\Path_F(E', e, f) \Leftrightarrow & (\exists E'' \subseteq (\IncE(v)
	\setminus e_\cA'))(E' = E'' \cup \{e, f\}) \\
	&\wedge e_1 \in E' \leftrightarrow \Big(\Big((e_1 = e \vee e_1 = f) \wedge
	(\exists e_2 \in E') (\Adj_F(e_1, e_2) \\
	&\wedge (\forall e_3 \in E') ((\neg
	e_2 = e_3) \to \neg \Adj_F(e_1, e_3)))\Big) \\
	\vee &~\Big(\neg (e_1 = e \vee e_1 = f) \wedge (\exists e_2 \in E')(\exists e_3
	\in E')\Big(\Adj_F(e_1, e_2) \\
	&\wedge \Adj_F(e_1, e_3) \wedge (\forall e_4 \in E') ((\neg (e_4 = e_2 \vee
	e_4 = e_3)) \\
	&\to \neg \Adj_F(e_1, e_4)) \Big) \Big) \Big)
\end{align*}
We are now ready to define the predicate for the ordering $\oriNB$.
\begin{align*}
	\oriNB(e, f) \Leftrightarrow \exists E_e \exists E_f (\Path_F(E_e, e_\cA, e)
	\wedge \Path_F(E_f, e_\cA, f) \wedge E_e \subset E_f)
\end{align*}

\subsubsection{Tree Decompositions for 3-Connected $k$-Outerplanar
Graphs}\label{appSecMSOL3CKOPGTD} 
We first show how to define that a spanning tree with edge set $F$ has bounded
face remember number $\nu$ in a 3-connected planar graph $G = (V, E)$, which 
completes the proof of Proposition \ref{propKOP3CSpTr}. Intuitively speaking,
this predicate checks that for each combination of a vertex and a face boundary
$FB$, the number edges, whose fundamental cycle uses both $v$ and some edge in
$FB$, is bounded by $\nu$.
\begin{align*}
	fr(V, E, F) \le \nu &\Leftrightarrow \forall v (\forall E_{FB} \subseteq
	E)\Big(\FaceB_{3}(E_{FB}) 
	\to (\forall e_1 \in E \setminus F)\cdots(\forall e_{\nu + 1} \in E \setminus
	F) \\
	\Big(\Big(&\bigwedge_{1 \le i \le \nu + 1}(\exists E_C \subseteq E)
	(\FundCyc(e_i, C_e) \wedge \neg (C_E \cap E_{FB} = \emptyset) \wedge \Inc(v,
	E_C))\Big) \\
	&\to \bigvee_{1 \le i < j \le \nu+1} e_i = e_j\Big)\Big)
\end{align*}
Next, we will define the edge sets $C(v, f_i)$, as used in the proof of Lemma
\ref{lemKOP3CTDDef}.
\begin{align*}
	E' = C(v, E_{FB}, F) &\Leftrightarrow e \in E' \leftrightarrow (\exists E_C
	\subseteq E) (\FundCyc(e, E_C) \\
	&\wedge \neg (E_C \cap E_{FB} = \emptyset)
	\wedge \Inc(v, E_C))
\end{align*}
We furthermore denote by $C(v, e, F)$ the union of the sets $C(v, f_i)$ and
$C(v, f_j)$ of the two faces $f_i$ and $f_j$, whose face boundaries contain $e$
(such that $e$ is incident to $v$).\par
We now define a predicate identifying a unique face boundary with lowest layer
number for each vertex.
\begin{align*}
	\Layer_i(E_{FB}) \Leftrightarrow & \FaceB_3(E_{FB}) \wedge \exists v (\Inc(v,
	E_{F_B}) \wedge v \in V_i) \\ 
	E' = E_{f_\ell}(v) \Leftrightarrow & (\exists e \in E')\Big(\Inc(v, e) \wedge
	\bigwedge_{i = 1,\ldots,k} \Big(v \in V_i \to \Big(\Big(\Layer_{i-1}(E')
	\\
	&\wedge \neg
 	((\exists f\exists E_f)(\Layer_{i-1}(E_f) \wedge f \in E_f \wedge \Inc(v, f)
 	\wedge \oriNB(f, e))))\Big) \\
 	&\vee \Big(\Layer_i(E') \wedge \neg (\exists E_f
 	(\Layer_{i-1}(E_f) \wedge \Inc(v, E_f))) \\
 	&\wedge \neg ((\exists f \exists E_f)(\Layer_i(E_f) \wedge f \in E_f \wedge
 	\Inc(v, f) \wedge \oriNB(f, e)))\Big)\Big)\Big)
\end{align*}
\par
We are now ready to define the $\Bag$-predicates of our tree
decomposition. Note that the bag type $\sigma$ can be defined in the same way as
for bounded degree $k$-outerplanar graphs, hence we refer to Appendix
\ref{appSecMSOLVRER} for the details. The types $\sigma_H$ can be defined using
the predicates given above.
We assume that we are given an arbitrary but fixed orientation on the edges as
described in the proof of Lemma \ref{lemKOP3CTDDef}.
\begin{align*}
	\Bag_{\sigma_H}(e, X) \Leftrightarrow & v \in X \leftrightarrow \head(v, e)
		\vee (\exists e' \in (C(v, e, F) \cup C(v, E_{f_\ell}(\head(e)), F)) \\
		& (\Inc(v, e') \wedge \forall w (\neg(v = w) \wedge \Inc(w, e')) \to \col(v)
		< \col(w))
\end{align*}
We can define the bag type $\sigma_T$ by replacing '$\head$' by '$\tail$' in the
above predicate.
\par
We now define the set of anchor edges $E_\cA$ and co-anchor edges $E_\cA'$.
For each vertex $v$ we need to find a face with lowest layer number $f_\ell$.
Let $e_{\ell_1}$ and $e_{\ell_2}$ denote the incident edges of $v$ bounding
$f_\ell$. Then, $e_{\ell_1}$ has to be contained in $E_\cA$ and $e_{\ell_2}$ in
$E_\cA'$. Note that this choice is arbitrary and that we have to choose
precisely one such face for each vertex in the graph.
\begin{align*}
	E' = E_\cA \Leftrightarrow &\forall v \exists e (e \in E' \wedge \Inc(v, e)
	\wedge e \in E_{f_\ell}(v) \\
	&\wedge \forall e'((\Inc(v, e') \wedge \neg e = e') \to \neg (e' \in E'))) \\
	E' = E_\cA' \Leftrightarrow &(\forall e \in E_\cA)\forall v \exists e'(e' \in
	E' \wedge \Inc(v, e) \wedge \Inc(v, e') \wedge e \in E_{f_\ell}(v) \wedge e'
	\in E_{f_\ell}(v) \\
	&\wedge \forall e'' ((\Inc(v, e'') \wedge \neg e'' = e') \to \neg (e'' \in
	E')))
\end{align*}
\par
We now turn to defining the $\Parent$-predicate and begin by defining the case
when a bag of type $\sigma$ is a bag of type $\sigma_T$.
\begin{align*}
	\Parent_{\sigma\sigma_T}(X, Y) \Leftrightarrow &(\exists e \in F)
	(\Bag_\sigma(e, X) \wedge \Bag_{\sigma_T}(e, Y)) \\
	\vee &(\exists e \in F)(\exists e_\ell \in E_{f_\ell}(\tail(e)) \cap
	\Inc(\tail(e)))(\Adj_F(e, e_\ell) \\
	&\wedge \Bag_{\sigma}(e_\ell, X) \wedge \Bag_{\sigma_T}(e, Y))
\end{align*}
Similarly, we can define the case when a bag of type $\sigma_H$ is the parent of
a bag of type $\sigma$.
\begin{align*}
	\Parent_{\sigma_H\sigma}(X, Y) \Leftrightarrow &(\exists e \in F)
	(\Bag_{\sigma_H}(e, X) \wedge \Bag_{\sigma}(e, Y)) \\
	\vee &(\exists e \in F)(\exists e_\ell \in E_{f_\ell}(\head(e)) \cap
	\Inc(\head(e)))(\Adj_F(e, e_\ell) \\
	&\wedge \Bag_{\sigma_H}(e, X) \wedge \Bag_{\sigma}(e_\ell, Y))
\end{align*}
We now consider edges between bags of type $\sigma_H/\sigma_T$. In the
following, we define the case when all bags involved are $\sigma_T$-bags and
note that the other cases can be defined by the obvious replacements. We first
define the outgoing edges of the $\sigma_T$-bag corresponding to the unique
incoming edge in the directed spanning tree $T = (V, F)$.
\begin{align*}
	\Parent_{\sigma_T\sigma_T}^I(X, Y) \Leftrightarrow & (\exists e^* \in
	F)(\exists e \in E) \Big(\Bag_{\sigma_T}(e^*, X) \wedge \Bag_{\sigma_T}(e, Y)
	\\
	&\wedge \tail(e^*) = \tail(e) \wedge (\oriNBA(e, e^*) \vee \oriNBA(e^*,
	e))\Big)
\end{align*}
We now define the rest of the edges. We denote by $e^*(v)$ the edge which
satisfies $e^* \in F \wedge \tail(e^*) = v$.
\begin{align*}
	\Parent_{\sigma_T\sigma_T}^R(X, Y) \Leftrightarrow \exists e \exists
	f&\Big(\tail(e) = \tail(f) \wedge \oriNBA(e, f) \\
	\wedge \Big(&(\Bag_{\sigma_T}(e, X) \wedge \Bag_{\sigma_T}(f, Y) \wedge
	\oriNB(e^*(\tail(e)), e)) \\
	\vee &(\Bag_{\sigma_T}(f, X) \wedge \Bag_{\sigma_T}(e, Y) \wedge \oriNB(f,
	e^*(\tail(f)))) \Big)\Big) \\
\end{align*}
Unifying all above defined predicates (plus the omitted similar cases) yields
the $\Parent(X, Y)$-predicate for our tree decomposition.

\subsection{Hierarchical Graph Decompositions for
$k$-Outerplanar Graphs}\label{appSecMSOLHGD} 
In this section we provide details for the predicates used in proofs of Section
\ref{secG3CC}. First we show how to define the parent-relation between blocks in
our hierarchical decomposition as explained in the proof of Lemma
\ref{lemConnTD3ConnDef}. We assume that we are given a graph $G = (V, E)$ with a
spanning tree $S = (V, F)$, which i  rooted at an (arbitrary) vertex $r \in V$.
\par
Let $\Block(X)$ denote a predicate which is true if and
only if a set $X \subseteq V$ is a block in the hierarchical decomposition of
$G$. $\Block(X)$ is definable by \cite{Cou99} (cf. also
Proposition \ref{propConnTD3DefUtil}). This predicate both encodes the cases of
the edges between 2-cuts and 3-blocks (see Proposition \ref{propC2B3dir}) and of edges
between 1-cuts and 2-blocks.
\begin{align*}
	\Parent_{\Block}(X, Y) \Leftrightarrow & (\Block(X) \wedge \Block(Y) \wedge (X
	\cap Y = X \vee X \cap Y = Y))
	\\
	\wedge &\Big((X \subset Y) \to (\forall v \in Y)(\exists x \in X)\forall
	E_{P_v}\exists E_{P_x} \\
	&(\Path(r, v, E_{P_v}) \wedge \Path(r, x, E_{P_x})
	\wedge E_{P_x} \subset E_{P_v})\Big)
	\\
	\wedge &\Big((Y \subset X) \to (\exists v \in Y)(\exists x \in X)\forall
	E_{P_v} \exists E_{P_x} \\
	&(\Path(r, v, E_{P_v}) \wedge \Path(r, x, E_{P_x})
	\wedge E_{P_v} \subset E_{P_x})\Big)
\end{align*}

\subsubsection{Defining a Cycle Block}\label{appSecMSOLHGDCyc}
We now show how to define the predicates for tree decompositions of a cycle
block $C = (W, E_C)$ as used in the proof of Proposition
\ref{propConnTD3DefCycle}.
First, we find the root $r_C \in W$ of the cycle.
\begin{align*}
	v = r_C \Leftrightarrow &(r \in W \wedge v = r) \vee \Big((\exists C_P \subset
	V)(\Parent_{\Block}(C_P, C) \\
	&\wedge \exists v ((\Bag_{\cC_1}(v, C_P) \vee
	\Bag_{\cC_2}(v, C_P)) \wedge v = r_C))\Big)
\end{align*}
Now we can define the predicate $\Bag_{Cyc}$ straightforwardly.
\begin{align*}
	\Bag_{Cyc}(e, X) \Leftrightarrow & \neg \Inc(e, r_C) \wedge (v \in X
	\leftrightarrow (\Inc(v, e) \vee v = r_C))
\end{align*}
Furthermore we can define the predicate $\Parent_{Cyc}(X, Y)$ as described in
the proof of Proposition \ref{propConnTD3DefCycle}.
\begin{align*}
	\Parent_{Cyc}(X, Y) \Leftrightarrow &\exists e \exists f \Big(\Bag_{Cyc}(e, X)
	\wedge \Bag_{Cyc}(f, Y) \wedge |X \cap Y| = 2 \\
	\wedge &(\exists Z \subseteq V) \Big(\Cycle_{\cB_3}(Z) \wedge \Inc(e, Z) \wedge
	\Inc(f, Z) \\
	\wedge &(\exists P_e \subseteq \IncE(Z)) (\exists P_f \subseteq \IncE(Z)) \\
	& (\dir{\Path}(r_C, \tail(e), P_e) \wedge \dir{\Path}(r_C, f, P_f) \wedge P_e
	\subset P_f) \Big)\Big)
\end{align*}

\subsubsection{Defining the Parent-predicate for $(\cT,
\cX)$}\label{appSecMSOLHGDPar} 
We now complete the proof of Lemma \ref{lemConnTD3ConnDef} by defining the
parent-relation in all bags of the resulting tree decomposition $(\cT, \cX)$ of
the graph $G$. During this step we also modify some of the $\Bag$-predicates,
since, as explained in the proof, a number of vertices might be added to each
bag in the tree decomposition. A vertex $v$ is added to a bag $X$, when it is a
member of a tree decomposition of a 2-connected 2-block or a 3-block and $v$ is
contained in the parent cut bag of $X$ in the hierarchical decomposition of $G$.
We show how to define such a predicate for an arbitrary case.
\begin{align*}
	\Bag_*'(X) \Leftrightarrow &(\exists X' \subseteq X) \Big(\Bag_*(X') \wedge v
	\in X \setminus X' \leftrightarrow \exists Y \exists Z \Big(X' \subseteq Z \\
	&\wedge (\kConn{2}_{\cB_2}(Z) \vee \kConn{3}_{\cB_3}(Z) \vee
	\Cycle_{\cB_3}(Z)) \\
	&\wedge \Parent_{\Block}(Y, Z) \wedge v \in Y \Big)\Big)
\end{align*}
In the following, we indicate that we refer to these modified bags by using the
notation '$\Bag'\ldots$' instead of '$\Bag\ldots$'. We define two cases: One, in
which a $\cC_1$- or $\cC_2$-block is a parent of a $\cB_3$-block and vice versa.
The cases for $\cC_1$- and $\cB_2$-blocks can be defined by the obvious
replacements. Note that the predicate $\Root_{\cB_3}$ can be defined
straightforwardly using the $\Bag_{\cB_3}(X)$- and
$\Parent_{\cB_3}(X, Y)$-predicates.
\begin{align*}
	\Parent_{\cC \cB_3}(X, Y) \Leftrightarrow &(\Bag_{\cC_1}'(X) \vee
	\Bag_{\cC_2}'(X)) \wedge \Bag_{\cB_3}'(Y) \wedge X \subseteq Y \wedge
	\Root_{\cB_3}(Y) \\
	\wedge &\exists Z (Y \subseteq Z \wedge \Parent_{\Block}(X, Z \setminus X)) \\
	\Parent_{\cB_3 \cC}(X, Y) \Leftrightarrow &\Bag_{\cB_3}'(X) \wedge
	(\Bag_{\cC_2}'(Y) \vee \Bag_{\cC_1}'(Y)) \wedge X \subseteq Y \\
	\wedge &\exists Z (X \subseteq Z \wedge \Parent_{\Block}(X \setminus Z, Z)) \\
	\wedge &\neg(\exists X' (\Parent_{\cB_3}(X', X) \wedge X' \subseteq Y))
\end{align*}
The $\Parent_{\cB\cC}(X, Y)$-predicate can now be defined as a unification of
all these cases.

\subsubsection{Defining Tree Decompositions for 3-Connected
3-Blocks}\label{appSecMSOLHGDSE} 
We now show how to define the predicates for defining the sets $\cS_E$ and
$\cR_V$ as outlined in the proof of Lemma \ref{lemSpTrSetsDef}. To shorten our
notation, we will use the symbol $E[B_3, \cS_E]$ instead of the term
'$\IncE(B_3) \cap \cS_E$'.
\begin{align*}
	\mbox{\ref{lemSpTrSetsDef1}} &~(\exists \cR_V \subseteq V)(\exists F \subseteq
	E) (\exists F' \subseteq E) (\exists \cS_E \subseteq E) (\cS_E = F \cup
	F') \ldots \\
	\mbox{\ref{lemSpTrSetsDef2}} &~ (\exists r_T \in V) (\dir{\Tree}(r_T, F))\ldots \\
	\mbox{\ref{lemSpTrSetsDef3}} &~ e \in F' \to \exists x \exists y (\neg x = y
	\wedge \Inc(x, e) \wedge \Inc(y, e) \wedge \neg e \in F \\
	&\wedge \exists
	X(\Bag_{\cC_2}(X) \wedge x \in X \wedge y \in X))\ldots \\
	\mbox{\ref{lemSpTrSetsDef4}} &~ (\forall B_3 \subseteq V)
	\Big(\kConn{3}_{\cB_3}(B_3) \to \Big(er(B_3, \IncE(B_3), E[B_3, \cS_E]) \le 2k
	\\
	&\wedge fr(B_3, \IncE(B_3), E[B_3, \cS_E]) \le k \Big)\Big) \\
	\mbox{\ref{lemSpTrSetsDef5}} &~ v \in \cR_V \to \exists X (\Bag_{\cC_2}(X)
	\wedge v \in X) \\
	\mbox{\ref{lemSpTrSetsDef6}} &~ (\forall B_3 \subseteq V)
	\Big(\kConn{3}_{\cB_3}(B_3) \to (\exists r_{B_3} \in \cR_V)
	\Big(\dir{\Tree}(r_{B_3}, B_3, E[B_3, \cS_E])\Big)\Big)
\end{align*}

\end{document}